\documentclass[journal]{IEEEtran}
\usepackage{amsmath,amsfonts}
\usepackage{algorithmic}
\usepackage{algorithm}
\usepackage{array}
\usepackage[caption=false,font=normalsize,labelfont=sf,textfont=sf]{subfig}
\usepackage{textcomp}
\usepackage{stfloats}
\usepackage{url}
\usepackage{color}
\usepackage{verbatim}
\usepackage{graphicx}
\usepackage{amsmath}
\usepackage{amssymb}
\usepackage{mathrsfs}
\usepackage{makecell}
\usepackage{threeparttable}
\usepackage[figuresright]{rotating}
\usepackage{soul} % 导入 soul 包
\usepackage{color, xcolor} % 颜色包，color 必须导入，xcolor 建议导入
\usepackage{booktabs}
\usepackage{bm}
\usepackage{multicol}
\usepackage{orcidlink}
\usepackage{tikz}
\makeatletter
\let\NAT@parse\undefined
\makeatother
\usepackage{hyperref}  %hyperref still needs to be put at the end!
\hypersetup{
	colorlinks=true,
	linkcolor=black,
	citecolor=black
}%取消参考文献彩色框

\newtheorem{lemma}{Lemma}
\newtheorem{corollary}{Corollary}
\newtheorem{theorem}{Theorem}
\newtheorem{definition}{Definition}
\newenvironment{proof}{{\noindent\it Proof:} }{\hfill $\square$\par}

\newtheorem{example}{Example}
\newtheorem{remark}{Remark}

\begin{document}
	
	\title{Combinatorial Constructions of Optimal Quaternary Additive Codes}
	
	%\author{Chaofeng Guan, Jingjie Lv, Yiting Liu, Zhi Ma~\IEEEmembership{}
	\author{Chaofeng Guan, Jingjie Lv, Gaojun Luo, Zhi Ma~\IEEEmembership{}
		% <-this % stops a space
		%\thanks{Chaofeng Guan, Yiting Liu, and Zhi Ma are with the Henan Key Laboratory of Network Cryptography Technology, Zhengzhou, 450001, China.  (e-mail: gcf2020@yeah.net, lyt9156@outlook.com, LWtglx2023@outlook.com).}
%		\thanks{Chaofeng Guan and Zhi Ma are with the Henan Key Laboratory of Network Cryptography Technology, Zhengzhou 450001, China  (e-mails: gcf2020@yeah.net, LWtglx2023@outlook.com).}
%			\thanks{Jingjie Lv is with the 
%			School of Electrical Engineering \& Intelligentization, Dongguan University of Technology, Dongguan, 523808,
%			China (email: juxianljj@163.com).}
%		\thanks{ Gaojun Luo is with the School of Mathematics, Nanjing University of Aeronautics and Astronautics, Nanjing, 210016, China (e-mail: gaojun\_luo@nuaa.edu.cn).}		
%		\thanks{Chaofeng Guan is supported by the National Key Research and Development Program of China under Grant 2021YFB3100100, and the National Natural Science Foundation of China under Grant U21A20428.
%			Zhi Ma is supported by the National Natural Science Foundation of China under Grant 61972413.
%			Gaojun Luo is supported by the Natural Science Foundation of Jiangsu Province Grant No. BK20230867, and the National Natural Science Foundation of China
%			Grant No. 12401690 and Postdoctoral Fellowship Program of CPSF Grant No. GZC20242234	}
%		\thanks{This work is supported by the National Natural Science Foundation of China
%			under Grant No. 61972413, U21A20428, 62002385, and the National Key R\&D Program of China under Grant No. 2021YFB3100100. }% <-this % stops a space
%		% The paper headers
	}
	%61972413,61901525,62002385,U21A20428

	\IEEEpubid{}

	\maketitle
	
	\begin{abstract}

 This paper aims to construct optimal quaternary additive codes with non-integer dimensions. Firstly, we propose combinatorial constructions of quaternary additive constant-weight codes, alongside additive generalized anticode construction. Subsequently, we propose generalized Construction X, which facilitates the construction of non-integer dimensional optimal additive codes from linear codes.
Then, we construct ten classes of optimal quaternary non-integer dimensional additive codes through these two methods. 
As an application, we also determine the optimal additive $[n,3.5,n-t]_4$ codes for all $t$ with variable $n$, except for $t=6,7,12$.
 %ten classes of them are additive Griesmer puncture-optimal codes or additive Griesmer codes. 
	\end{abstract}
	
	\begin{IEEEkeywords}
Quaternary additive code, additive constant-weight code, additive generalized anticode construction, generalized Construction X.
	\end{IEEEkeywords}
	\section{Introduction}

\IEEEPARstart{A}{dditive} codes are closed under vector addition but not necessarily under scalar multiplication.
All linear codes are additive, but additive codes are not necessarily linear. 
Therefore, the parameters of additive codes may be better than optimal linear codes.
%The algebraic structure of additive codes is more complex than that of linear codes, the parameters of additive codes may be better than optimal linear ones.
Moreover, additive codes are an interesting class of error-correcting codes and have significant applications in quantum information \cite{calderbank1998quantum,ketkar2006nonbinary}, computer memory systems \cite{chen1984error,chen1991fault,chen1992symbol}, deep space communication \cite{hattori1998subspace}, protection against side-channel attacks \cite{shi2023additive,shi2022additive}, and distance-regular graphs \cite{Shi2018AND}.%, kim2017secret,annamalai2021additive

In \cite{blokhuis2004small}, Blokhuis and Brouwer determined the parameters of optimal quaternary additive codes of lengths up to 12.
Later in \cite{bierbrauer2009short,Danielsen2009,bierbrauer2010geometric,bierbrauer2015nonexistence,bierbrauer2019additive}, much work has been performed on quaternary additive codes with small lengths, resulting in a general determination of the parameters of optimal additive codes of lengths up to 15.
In \cite{bierbrauer2021optimal}, Bierbrauer et al. determined the parameters of all 2.5-dimensional optimal quaternary additive codes by a projective geometry approach. 
However, as the dimensions increase, the number of involved optimal additive codes increases exponentially.
For instance, for $k=2.5$, the parameters of optimal quaternary 2.5-dimensional additive codes can be determined by only three non-trivial key codes, while the problem is much more complicated for $k\ge 3.5$. 
In \cite{guo2017construction,Guan2023SomeGQ,kurz2024computer}, many optimal quaternary 3.5-dimensional additive codes were obtained by computer-supported methods. However, this work could only partially determine the optimal quaternary 3.5-dimensional additive codes of lengths from 28 to 254. A solution to the problem of constructing optimal quaternary 3.5-dimensional or higher-dimensional additive codes requires more efficient methods and theories.
%Since the optimal additive code period involved increases exponentially with the increase of dimension, the parameters of the 3.5-dimensional optimal quaternary additive codes are only partially determined .

%obtained many additive codes with parameters better than the best quaternary linear codes in \cite{Grassltable}, and partially determined the parameters of optimal quaternary additive 3.5-dimensional codes with lengths from 28 to 254.

The Griesmer bound is essential in determining the performance of linear codes, and the following is the generalization of the Griesmer bound for additive codes.
\begin{lemma}\label{Griesmer_Bound}(Additive Griesmer Bound, \cite{Guan2023SomeGQ})
	If $C_a$ is a quaternary additive $[n,k,d]_4$ code with $k\ge1$, then \begin{equation}\label{AGB}
		3n \geq g(2k, 2d) \triangleq \sum\limits_{i=0}^{2k-1}\left\lceil\frac{d}{2^{i-1}}\right\rceil.
	\end{equation} 
\end{lemma}

Follow the work of Bierbrauer \cite{bierbrauer2021optimal}, we also prefer to work with the \textit{species} $t=n-d$ instead of minimum distance $d$. 
	%To distinguish optimal scenarios concerning the above bound, 
 We define $C_{n,t}$ as an additive $[n,k,n-t]_4$ code, where $n$, $k$, and $t$ are positive integers. 
	%$C_{n,t}$ can represent all additive codes with dimension $k$. 
	As in the linear case, an $[n,k,n-t]_4$ additive code satisfying the additive Griesmer bound with equality is called anan additive \textit{Griesmer code}. 

For fixed dimension $k$, we call $C_{n,t}$ an additive \textit{distance-optimal (DO)} code if $t$ is the smallest number for which there exists an additive $[n,k,n-t]_4$ code. 
Similarly, we call $C_{n,t}$ an additive \textit{puncture-optimal (PO)}  code if $n$ is the maximum length for which there exists an additive $[n,k,n-t]_4$ code.
%if $t$ is the minimum number such that additive $[n,k,n-t]$ code exists, then we call $C_{n,t}$ an additive \textit{distance-optimal (DO)}  code. 
%Similarly, if $n$ is the maximal length such that an additive $[n,k,n-t]$ code exists, then we call $C_{n,t}$ an additive \textit{puncture-optimal (PO)}  code. 
If the optimality type of $C_{n,t}$ can be determined directly via Lemma \ref{Griesmer_Bound}, then we call $C_{n,t}$ an additive \textit{Griesmer distance-optimal (GDO)} or \textit{Griesmer puncture-optimal (GPO)}  code.
	% \begin{equation}\label{griesmer_opt}
	% 	%g(2k, 2(n-t)) \le 
	% 	3n  < g(2k, 2(n-t)+2),
	% \end{equation}
	% then we call $C_{n,t}$ an additive \textit{Griesmer distance-optimal (GDO) code}. 
	% For fixed $n$ and $k$, if there is no GDO code, then we refer to the code with minimum $t$ that exists as an additive \textit{distance-optimal (DO) code}.
%	For a fixed $t$, if $C_{n,t}$  satisfies 
%	\begin{equation}\label{griesmer_opt}
%		%g(2k, 2(n-t)) \le 
%		3n  < g(2k+1, 2(n-t)),
%	\end{equation}
%	then we call $C_{n,t}$ an additive \textbf{Griesmer dimensional-optimal (GDIO) code}.; if no such code $C_{n,t}$ exists, then we refer to the code with maximum $t$ that exists as an additive \textbf{distance-optimal (DO) code}.
	% For fixed $k$ and $t$, if $C_{n,t}$ satisfies 		 
	% \begin{equation}\label{griesmer_rd_opt}
	% 	%g(2k, 2(n-t)) \le 
	% 	3n+3  < g(2k, 2(n-t)+2),
	% \end{equation} i.e., additive $[n+1,k,n-t+1]_4$ code violates Griesmer bound, then we call $C_{n,t}$  an additive \textit{Griesmer puncture-optimal (GPO) code}.
	% For fixed $k$ and $t$, if there is no GPO code exists, then we call the longest GDO (or DO) code an additive \textit{puncture-optimal (PO)} code.
	This is, all optimal additive GDO (or DO) codes can be punctured from additive GPO (or PO) codes,
	 otherwise they are also GPO (or PO).
	Therefore, for a specific dimension $k$, when we want to construct optimal additive code $C_{n,t}$ for all $t$, we only need to construct all $k$-dimensional GPO and PO codes.
	These facts suggest that GPO and PO codes are the most critical optimal codes.

	The main objective of this paper is to construct optimal quaternary additive codes with non-integer dimensions using combinatorial methods.
	As a result,  many classes of optimal quaternary additive codes are obtained, and their weight distributions are also determined.
%	  It is worth mentioning that most of our additive codes outperform the corresponding optimal linear codes in terms of the number of codewords for given lengths and minimum distances, with $0.5$ more dimensions, i.e., twice as many codewords.
	The contributions of this paper can be summarized as follows.
	\begin{enumerate}
		\item 
We present three methods for constructing a class of quaternary constant-weight codes with parameters $[2^{k}-1,\frac{k}{2},3\cdot2^{k-2}]_4$ and prove that all their codewords satisfy the condition that $1$, $w$, and $w^2$ appear $2^{k-2}$ times, see Lemmas \ref{APS_construction}, \ref{iterating} and \ref{L_camba}.
In addition, we also obtain two classes of derived codes and determine their weight distributions, see Lemma \ref{aug_Simplex}. 

\item We propose an additive generalized anticode construction of optimal additive codes based on a combinatorial approach, see Lemma \ref{anticode construction} and Theorem  \ref{anticode_k_1_k_2}.
With the help of an equivalent division method, we improve the parameters of the resulting additive codes, see Theorem  \ref{anticode_k_1_k_2_1/3}. The weight distributions of related optimal additive codes are also completely determined.
\item We provide a generalized Construction X.
By choosing a special auxiliary $[n,k^{\prime},d]_4$ code and combining it with two $[n_1,k,d_1]_4$ and $[n_2,k,d_2]_4$ codes, an additive $[n_1+n_2+n,k,>d_1+d_2+d]_4$ code can be generated, as
shown in Theorem \ref{Generalized_X_Construction}.
In addition to directly constructing optimal additive codes, it can also establish a connection between optimal additive codes with non-integer dimensions and linear codes, see Lemma \ref{Generalized_X_Construction_Corollary}.

\item We give a method to determine the existence of quaternary additive codes using binary linear codes (Lemma \ref{L_anti_bound}). 
A tighter bound concerning the Griesmer bound for quaternary additive codes with short lengths is obtained. 
By the optimal additive codes obtained in this paper, we ultimately determine optimal additive $[n,3.5,n-t]_4$ codes for all $t$ with variable $n$, except for $t=6,7,12$, see Lemma \ref{Additive code combinations} and Table \ref{3.5_additive}.
	\end{enumerate}

	%\ref{Sec III ASEP}, \ref{Sec IV anti}\ref{Sec V ASEP}\ref{Sec_3.5_4.5}
		This paper is organized as follows. In Section \ref{Sec II}, we recall some preliminary results needed in this paper. Sections \ref{Sec III ASEP} and \ref{Sec IV anti} are devoted to additive constant-weight codes and related generalized anticode construction.
Section \ref{Sec V ASEP} presents a generalized Construction X. 		
In Section \ref{Sec_3.5_4.5}, we determine the parameters of all the quaternary 3.5-dimensional GPO and PO codes, except for three. Finally, Section \ref{VII Dis conclu}  concludes this paper.

	\section{Preliminaries}\label{Sec II}
	This section presents the fundamentals of additive codes. 
	For more details, refer to the standard textbooks \cite{huffman2010fundamentals,bierbrauer2017introduction,huffman2021concise}.
	
	%\subsection{Quaternary additive codes}
	Denote the finite fields of order $2$ and $4$ by respectively $F_2=\{0,1\}$ and $F_4 = \{0, 1, w, w^2\}$ with $w^2 +w + 1 = 0$.
	For $\bm{u }=(u_{0},\ldots, u_{n-1})\in F_{2}^{n}$, the Hamming weight of $\bm{u}$ is 
	$\mathbf{w}(\bm{u})=|\left\{i \mid u_{i} \neq0, 0 \leq i \leq n-1\right\}|$. 
	For $\bm{v}=(v_{0},\ldots, v_{2n-1}) \in F_{2}^{2n}$,  
	the symplectic weight of $\bm{v}$ is 
	$\mathbf{w}_{s}(\bm{v})=|\left\{i \mid (v_{i}, v_{n+i}) \neq(0,0), 0 \leq i \leq n-1 \right\}|$.
		For $\bm{v}_1,\bm{v}_2 \in \mathbb{F}_q^{2n}$,  the dot and symplectic inner products of them are defined as $\langle\bm{v}_1,\bm{v}_2\rangle_{e}=\sum_{i=0}^{2 n-1} v_{1,i} v_{2,i}$ and 
	$\langle\bm{v}_1,\bm{v}_2\rangle_{s}=\sum_{i=0}^{n-1}\left(v_{1,i} v_{2,n+i}-v_{1,n+i} v_{2,i}\right)$, respectively. 
	Let $q=2$ or $4$, a code $C_l$ is linear over $F_q$ if it is a linear subspace of $F_q^n$;   
	if such a code has dimension $k$, minimum Hamming distance $d$, then $C_l$ is denoted by $[n,k,d]^l_q$.
	If $C_s$ is a $[2n,k]_2^l$ code with the minimum symplectic distance $d_s$, then $C_s$ is denoted as $[2n,k,d_s]^s_{2}$.
	An additive $ ( n,K,d ) _4$ code $C_a$ is an additive subgroup of $F_4^n$ with $K$ elements whose minimum Hamming weight is $d$, which is an $F_2$-linear $F_4$-code. Writing dimension $k=\log_4 K$, then $C_a$ also can be denoted as $[n,k,d]_4$, where $k$ is not necessarily integer.
Throughout this paper, to avoid confusion, we use different subscripts to represent the three types of codes involved in this paper: \( C_l \) for linear code with Hamming weight, \( C_s \) for linear code with symplectic weight, and \( C_a \) for additive code with Hamming weight.

%	The ratio $\frac{k}{n}$ of a code is its information rate.
%	 We say an additive $[n, k + 0.5, d]_4$ code for integer $k$ outperforms the linear codes if no linear $[n, k+1, d]^l_4$ code exists.
%	A code $C_a$ is said to be an additive code over $F_{4}$ if it is an additive subgroup of $F^n_{4}$.
%	If $C_a$ has dimension $k_a$, minimum Hamming distance $d$, then $C_a$ is denoted as $[n,k_a,d]_{4}$.
%		The following lemma gives sufficient conditions for additive codes to outperform linear codes. 
%	\begin{lemma}\label{outperform_linear}
%		Let $C_{n,t}$ be an additive Griesmer code with parameters $[n,k,n-t]_4$, where $k$ is odd.
%		If $C_{n,t}$ is a both Griesmer dimensional and puncture optimal  Quaternary additive GPO and Griesmer codes with non-integer dimensions outperform the linear codes.
%\end{lemma}
%\begin{proof}
%Let $C_{n,t}$ be an additive Griesmer code with parameters $[n,k,n-t]_4$. Since $g(2k,2(n-t))=3n$, $[n,k+0.5,n-t]_4$ code do not exists. 
%Moreover, according to the definition of GPO code,  all Griesmer codes are GPO codes or can be punctured from GPO codes. 
%Therefore, the GPO code has parameters $[n+ \Delta,k,d+ \Delta  ]_4$, where $\Delta$ is a non-negative integer. 
%Obviously, since $[n,k+0.5,n-t]_4$ code do not exists, $[n+ \Delta,k+0.5,d+ \Delta  ]_4$ code also do not exists.
%	\end{proof}

	Define a map $\Phi$ from $F_2^{2n}$ to $F_4^n$ as $\Phi(\bm{v})=(v_0+w\cdot v_n,v_1+w\cdot v_{n+1},\cdots,v_{n-1}+w\cdot v_{2n-1})$, where $w$ generates $F_{4}$ over $F_{2}$.
Then it is easy to see that $\Phi$ is a bijection and $\mathbf{w}(\Phi(\bm{v}))=\mathbf{w}_{s}(\bm{v})$.
Let $G_s$ be a generator matrix of $C_s$. Then, $\Phi(G_s)$ generates an additive code $C_a$ with parameters $[n,\frac{k}{2},d_s]_{4}$. 
 {The \textit{additive form} of $F_4$-vector $\bm{m}=(m_0,m_1,\ldots,m_{k-1})^T$ is $(m_0,w*m_0,m_1,w*m_1,\ldots,m_{k-1},w*m_{k-1})^T$. 
Let $F_{4_a}^k$ denote the vector space consisting of all vectors in $F_4^k$ written in the additive form.}

The extended code of $C_a$ is obtained by adding an appropriate extra coordinate to each codeword of $C_a$ such that the sum of the coordinates of all extended codewords is zero.
The weight distribution of a code $C_a$ is defined by $1+A_1z+A_2z^2+\cdots+A_nz^n$, where $A_i$ is the number of codewords of weight $i$, $1\le i \le n$. 
We call $C_a$ a $t$-weight code if the number of nonzero $A_i$ in the sequence $(A_1, A_2, \cdots , A_n)$ equals $t$.

A $k$-dimensional vector space over $F_2$ can be viewed as $PG(k-1,2)$. 
In $PG(k-1,2)$, one-dimensional subspaces are called points, and the two-dimensional subspaces are called lines. Let $\mathcal{P}$ and $\mathcal{L}$ represent the set of all points and lines in $PG(k-1,2)$, respectively.
The number of points in $PG(k-1,2)$ is $|\mathcal{P} |=2^k-1 $ and the number of lines is $|\mathcal{L} |=\frac{|\mathcal{P} |\cdot(|\mathcal{P}|-1)}{6}$.
A set of all points generates the well-known Simplex code with parameters
$[2^k-1,k,2^{k-1}]_2$, which is also a constant-weight cyclic code\footnote{Binary cyclic Simplex code can be generated by a polynomial with defining set $T=Z_n\setminus C_{1}$.}. 
In \cite{blokhuis2004small}, Blokhuis et al. took quaternary additive code as multisets of lines in $PG(k-1, 2)$, this is, viewed an $F_4$-additive code as an $F_2$-code with symplectic weight. Bierbrauer et al. proposed the parameters $[\frac{(2^k-1)(2^{k-1}-1)}{3}, \frac{k}{2}, 2^{k-2}(2^{k-1}-1)]_4$ of additive Simplex code in \cite{bierbrauer2009short}. 
It is easy to see that the additive Simplex code meets the additive Griesmer bound with equality.

At the end of this section, we introduce a particular lengthening method for additive codes, called Construction X, which is an efficient technique for constructing optimal additive codes in \cite{Guan2023SomeGQ}.
\begin{lemma}(Additive Construction X, \cite{Guan2023SomeGQ})\label{construction_X}
	Assume that $C_{a}$ is an additive $[n,k_1,d_1]_{4}$ code with subcode $[n,k_2,d_2]_{4}$, where $d_2 > d_1$.
	Let $C_{au}$ be an additive $[l,k_1-k_2,\delta ]_{4}$ code. Then, there exists an additive $[n+l,k_1,\min\{\delta+d_1,d_2\} ]_{4}$ code.
\end{lemma}

\begin{proof}
 	Since $C_{a}$ has an $[n,k_2,d_2]_{4}$ subcode, $C_a$ can be divided into two disjoint parts, subcodes $C_{s_1}$ and $C_{s_2}$ with parameters $[n,k_2,d_2]_{4}$ and $[n,k_1-k_2,\ge d_1]_{4}$, respectively. Obviously, we have  $C_a=C_{s_1}+C_{s_2}$, and it is $C_{s_2}$ or the combination of $C_{s_1}$ and $C_{s_2}$ decrease the distance of $C_a$ from $d_2$ to $d_1$.  
 	An economic approach to augment the distance of $C_a$ to $d_2$ is horizontal join an auxiliary code $C_{au}$ with $C_{s_2}$. % so that the result code $C_{ax}$ has smaller length $n+l$. 
 	%Let $C_{ax}$ denote the result code. 
 {	Let $G_{s_1}$, $G_{s_2}$  and $G_{au}$ denote generator matrices of $C_{s_1}$, $C_{s_2}$ and $C_{au}$, respectively. Then, additive Construction X can be described as the following matrix form:
	\begin{equation}
		G_{ax}=\left ( \begin{matrix}
			G_{s_1} & \mathbf{0}_{k_2\times l}\\
			G_{s_2}& G_{au}
		\end{matrix} \right ) ,
	\end{equation}
	where $\mathbf{0}_{k_2\times l}$ is a $k_2\times l$ zero matrix. 		
	Obviously, $G_{ax}$ generates an additive $[n+l,k_1,\min\{\delta+d_1,d_2\} ]_{4}$ code.}
\end{proof}
%\begin{lemma}
%		
%		Let $C$ be a linear $[n,k,d]_q^l$ code with generator matrix $G$. Then, $\Phi(G|A_fG)$ generates an additive $[n,k,d]_{q^2}$ code. 
%
%		\begin{equation}
%			A=\begin{pmatrix}
%				\bm{r_0} + w \cdot \bm{r_1}       \\
%				\bm{r_1} +w \cdot \bm{r_2}        \\
%				\vdots                            \\
%				\bm{r_{k-2}}  +w \cdot \bm{r_{k-1}}\\
%				\bm{r_{k-1}}  +w \cdot  \sum\limits_{i \in \mathcal{S}_f} \bm{r_i}
%			\end{pmatrix},
%		\end{equation}
%	where $\bm{r_i}$ is 
%\end{lemma}
%	\begin{proof}
%		Denote $G$ as $G=\begin{pmatrix}
%			\bm{r_0}\\ \bm{r_{1}}\\\vdots\\\bm{r_{k-1}} 
%		\end{pmatrix}$
%	\end{proof}

\section{Additive symbol equal probability codes}\label{Sec III ASEP}

	Given $\alpha\in F_4$, let $\bm{\alpha_{k,n}}$ denote the $k\times n$ matrix such that each entry is $\alpha$, and the notation $\bm{\alpha_{k,n}}$ is simplified as $\bm{\alpha_{n}}$ if $k=1$, and $\bm{\alpha_{k}}^T$ if $n=1$. 
	 Let $\bm{c}=(c_0,c_1,\cdots,c_{n-1}) \in F_4^n$ and $\mathbf{w}_{\alpha}(\bm{c})=|\{i\mid c_i= \alpha, 0 \leq i \leq n-1 \}|$. 
 {	 Define cyclic shift operator $\tau$ as $\tau(\bm{c})=(c_{n-1},c_0,\cdots,c_{n-2})$.}
	Let $[n,k^{\prime},d^{\prime}]_4\subset [n,k,d]_4$ denote that the additive $[n,k,d]_4$ code has an $[n,k^{\prime},d^{\prime}]_4$ subcode.

\begin{definition}
	Let $C_a$ be a quaternary additive code generated by $G_s$. If any codeword $\bm{c}$ of $C_a$ satisfies
	\begin{equation}
		\mathbf{w}_1(\bm{c})=\mathbf{w}_w(\bm{c})=\mathbf{w}_{w^2}(\bm{c}),
	\end{equation}
	then we call $C_a$ an \textit{additive symbol equal probability (ASEP) code}. 
\end{definition}

\begin{lemma}\label{APS_construction}
 {Let $f(x)=f_0+\cdots+f_{k-2}x^{k-2} +f_{k-1}x^{k-1}+ x^k$ be a monic irreducible polynomial\footnote{The property of finite fields guarantees that such polynomials can be found for any $k$ and $q$. A simple way is to choose the generator polynomial of $q$-ary ($n-k$)-dimensional cyclic Hamming code.} with degree $k$ over $F_2[x]$, and $\mathcal{S}_f=\{f_0,f_{1},\ldots,f_{k-1}\}$. 
	Let $S_k=\begin{pmatrix}
	\bm{r_0}\\ \bm{r_{1}}\\\vdots\\\bm{r_{k-1}} 
\end{pmatrix}$ be a generator matrix of binary $k$-dimensional Simplex code $C_{S_k}$, where $\bm{r_i}$ is row vector, $k\ge3$, $0\le i\le k-1$. 
Then, we have the following results.}

 {(1) ASEP code $C_{A_{\frac{k}{2}}}$ with parameters $[2^{k}-1,\frac{k}{2} ,3\cdot2^{k-2}]_4$ can be generated by $A_{\frac{k}{2}}$.}

 {(2) If $S_k$ generates a cyclic Simplex code $C_{S_k}$ and satisfies $\bm{r_{i+1}}=\tau(\bm{r_i})$, then $A_{\frac{k}{2}}^c$ also generates $C_{A_{\frac{k}{2}}}$.
\begin{equation}
	A_{\frac{k}{2}}=\begin{pmatrix}
		\bm{r_0} + w \cdot \bm{r_1}       \\
		\bm{r_1} +w \cdot \bm{r_2}        \\
		\vdots                            \\
		\bm{r_{k-2}}  +w \cdot \bm{r_{k-1}}\\
		\bm{r_{k-1}}  +w \cdot  \sum\limits_{i \in \mathcal{S}_f} \bm{r_i}
	\end{pmatrix},
	A_{\frac{k}{2}}^c=\begin{pmatrix}
		\bm{r_0} + w \cdot \bm{r_1}       \\
		\bm{r_1} +w \cdot \bm{r_2}        \\
		\vdots                            \\
		\bm{r_{k-2}}  +w \cdot \bm{r_{k-1}}\\
		\bm{r_{k-1}}  +w \cdot  \tau(\bm{r_{k-1}})
	\end{pmatrix}.
\end{equation}}

 {(3) The ASEP code $C_{A_{\frac{k}{2}}}$ is an additive both GPO and Griesmer code.}
\end{lemma}
\begin{proof}	
 {Let $A_f=\begin{pmatrix}
	& 1    &        &           \\
	&      & \ddots &           \\
	&      &        &      1     \\
	f_0 & f_1 & \cdots & f_{k-1} 
\end{pmatrix}$ be a $k\times k$ square matrix over $F_q$, where the empties indicate zero. Let $G_s=(S_k,A_fS_k)=	\left( 
\begin{array}{c|c}
	\bm{r_0}& \bm{r_1}\\ \bm{r_{1}}&\bm{r_2}\\\vdots\\\bm{r_{k-1}} &\sum\limits_{i \in \mathcal{S}_f} \bm{r_i}
\end{array}
\right)$. 
By Lemma 2.5 in \cite{ling2010generalization}, $f(x)$ is a characteristic polynomial of $A_f$, so $Rank(A_f)=k$ and $A_f$ has no eigenvalues in $F_2$. 
Let $C_s$ be the code generated $G_s$ and $\bm{c_I}$ be an $F_2$-vector of length $k$. Then, $\bm{c}=\bm{c_I}(S_k,A_fS_k)=(\bm{c_l},\bm{c_r})$ can represent all codewords of $C_s$. 
Since, $f(x)$ is an irreducible polynomial, for all nonzero $\bm{c_I}$, we have $\bm{c_l}-\bm{c_r}=\bm{c_I}(I_k-A_f)S_k\ne \bm{0}$.
Therefore, $\bm{c_l}+\bm{c_r}$ is also a nonzero codeword of $C_{S_k}$.
With Lemma 2.4 in \cite{ling2010generalization},  $\mathbf{w}_s(\bm{c})=(\mathbf{w}(\bm{c_l})+\mathbf{w}(\bm{c_r})+\mathbf{w}(\bm{c_l}+\bm{c_r}))/2\ge 3\cdot2^{k-2}$, so $C_s$ has parameters $[2^{k+1}-2,\frac{k}{2} ,3\cdot2^{k-2}]_2^s$, i.e., $A_{\frac{k}{2}}=\Phi(G_s)$ generates an additive $[2^{k}-1,\frac{k}{2} ,3\cdot2^{k-2}]_4$ code.}

 {In addition, if $C_{S_k}$ is a cyclic Simplex code with generator polynomial $g(x)$, and $S_k$ satisfies $\bm{r_{i+1}}=\tau(\bm{r_i})$, then the code generated by $\Phi^{-1}(A_{\frac{k}{2}}^c)$ can be viewed as a quasi-cyclic code with generator $(g(x),xg(x))$. By Theorem 1 in \cite{Guan2023SomeGQ}, $(g(x), xg(x))$ generates a quasi-cyclic code $C_{qc}$ with parameters $[2^{k+1}-2,\frac{k}{2} ,3\cdot2^{k-2}]_2^s$. 	
Therefore, $A_{\frac{k}{2}}^c$ also generates an additive $[2^{k}-1,\frac{k}{2} ,3\cdot2^{k-2}]_4$ code.}
	
%	A generator
%	matrix of $\Phi(Cs)$ can be written as: 
%		\begin{equation}
%		\Phi(G_s)=\begin{pmatrix}
%			\bm{r_1}+w \cdot \bm{r_1} +w \cdot \bm{r_2} \\  \bm{r_2}+w \cdot \bm{r_2} +w \cdot \bm{r_3}  \\ \vdots   \\ \bm{r_k} +w \cdot \bm{r_k} +w \cdot \bm{r_1}
%		\end{pmatrix}
%		=w^2\cdot \begin{pmatrix}
%			\bm{r_1} +w \cdot \bm{r_2} \\  \bm{r_2} +w \cdot \bm{r_3}  \\ \vdots   \\ \bm{r_k}  +w \cdot \bm{r_1}
%		\end{pmatrix}
%		=w^2\cdot A_{\frac{k}{2}}.
%	\end{equation}
	
 {	Furthermore, since $k$-dimensional Simplex code $C_{S_k}$ is a constant-weight linear $[2^k-1,k,2^{k-1}]^l_2$ code , for arbitrary nonzero codewords $\bm{c_a}$ and $\bm{c_b}$ ($\bm{c_a}\ne \bm{c_b}$) in $C_{S_k}$, we have $\mathbf{w}_0(\bm{c_a}+\bm{c_b})=2^{k-1}-1$ and $\mathbf{w}_1(\bm{c_a}+\bm{c_b})=2^{k-1}$. 
	Obviously, both the left and right sides of $\Phi^{-1}(A_{\frac{k}{2}})$ and $\Phi^{-1}(A_{\frac{k}{2}}^c)$ can generate $C_{S_k}$. Also, the two sides of their nonzero codewords are always different.
	Therefore, $\bm{c_a}+w\bm{c_b}$ can represent all codewords of $C_{A_{\frac{k}{2}}}$, and we have 
	\begin{equation}
		\begin{array}{l} 
			\left\{  \begin{array}{l}
			\mathbf{w}_0(\bm{c_a}+w\bm{c_b})=2^{k-2}-1. \\
			\mathbf{w}_0(\bm{c_a}+w\bm{c_b})+\mathbf{w}_w(\bm{c_a}+w\bm{c_b})=2^{k-1}-1. \\
			\mathbf{w}_1(\bm{c_a}+w\bm{c_b})+\mathbf{w}_w(\bm{c_a}+w\bm{c_b})=2^{k-1}. \\ 
			\mathbf{w}_w(\bm{c_a}+w\bm{c_b})+\mathbf{w}_{w^2}(\bm{c_a}+w\bm{c_b})=2^{k-1}. 
		\end{array}\right. 
	\\
	\Rightarrow 
	\left\{  \begin{array}{l}
		\mathbf{w}_1(\bm{c_a}+w\bm{c_b})=2^{k-2}. \\
		\mathbf{w}_w(\bm{c_a}+w\bm{c_b})=2^{k-2}. \\ 
		\mathbf{w}_{w^2}(\bm{c_a}+w\bm{c_b})=2^{k-2}. 
	\end{array}\right. 
	\end{array} 
	\end{equation}}
	
 {	Thus, the additive $[2^{k}-1,\frac{k}{2} ,3\cdot2^{k-2}]_4$ codes generated by  $A_{\frac{k}{2}}$ and $A_{\frac{k}{2}}^c$ are both ASEP codes. }
	
 {		Finally, since $g(k,3\cdot2^{k-1})=3 (2^{k}-1)$ and $g(k,3\cdot2^{k-1}+2)=3\cdot2^{k}+k-2>3\cdot2^{k}$, by Lemma \ref{Griesmer_Bound}, additive $[2^{k}-1,\frac{k}{2} ,3\cdot2^{k-2}]_4$ code is an additive both GPO and Griesmer code.}
		%$\sum\limits_{i=0}^{k-1}\left\lceil\frac{3\cdot2^{k-2}}{2^{i-1}}\right\rceil=3\cdot\sum\limits_{i=0}^{k-1}2^{k-i-1}$ and 
\end{proof}

\begin{remark}
	Remarkably, Bierbrauer et al. \cite{Bierbrauer2023AnAP} also recently investigated geometrically optimal additive codes with parameters $[2^{k}-1,\frac{k}{2} ,3\cdot2^{k-2}]_4$ and proved they are 3-cover.
	Here, Lemma \ref{APS_construction} can be regarded as another construction of these optimal additive codes from a different point of view. 
\end{remark}
%Lemma \ref{APS_construction} is a convenient way to construct all APS codes, which can be regarded as a particular combination of points in $PG(k-1,2)$. %Therefore, we can regard APS codes as consisting of additive points in a generalized sense. 

\begin{example}
 {	When $k=3$, we can easily get a 1.5-dimensional ASEP code $C_{A_{1.5}}$ from 3-dimensional Simplex code by Lemma \ref{APS_construction}.  
	Here, we show the details.}
	
 {	Firstly, the generator matrix $S_3$ of binary $3$-dimensional Simplex code may be chosen as
	\begin{equation}
		S_3=\left(\begin{array}{ccccccc}
			1 & 0 & 0 & 1 & 1 & 0 & 1 \\
			0 & 1 & 0 & 1 & 0 & 1 & 1 \\
			0 & 0 & 1 & 0 & 1 & 1 & 1
		\end{array}\right).
	\end{equation}}

 {	By (1) of Lemma \ref{APS_construction}, we can choose a monic irreducible polynomial $f(x)=x^3 + x + 1$, which corresponding to a non-singular matrix
	$A_f=\left(\begin{array}{ccc}
		0&1&0\\
		0&0&1\\
		1& 1& 0
	\end{array}\right)$. $A_f$ has no eigenvalues in $F_2$.
	Then, we have 
	\begin{equation}
					\setlength{\arraycolsep}{0.5pt}
		A_{1.5}=\Phi(S,A_fS)=\left(\begin{array}{ccccccc}
			1&w&0&w^2&1& w& w^2\\
			0&1&w&1&w& w^2&w^2\\
			w& w& 1& 0& w^2& w^2& 1
		\end{array}\right).\end{equation}
	It is easy to verify that $A_{1.5}$ generates an ASEP code $C_{1.5}$ with parameters $[7,1.5,6]_4$. }
%	\begin{equation}\label{1.5ASEP}
%		A_{1.5}=\begin{pmatrix}
%			1 & w & 1 & w^2 & w^2 & w   & 0   \\
%			w & 0 & 1 & w   & 1   & w^2 & w^2 \\
%			0 & 1 & w & 1   & w^2 & w^2 & w
%		\end{pmatrix}.
%	\end{equation}
%	One can check that $A_{1.5}$ indeed generates an ASEP code with parameters $[7,1.5,6]_4$. 
\end{example}

For ASEP Griesmer codes with integer dimensions, we prefer a more straightforward approach to construct them, as the following lemma.

\begin{lemma}\label{divided_Simplex} 
	Let $G_l$ be a generator matrix of the quaternary $l$-dimensional Simplex code, where $l\ge2$.  
	Then, the following matrix $A_{l}$ generates an ASEP Griesmer code $C_{A_l}$ with parameters $[2^{2l}-1,l,3\cdot2^{2l-2}]_4$.  
	\begin{equation}\label{integer_ASEP}
		A_l=\begin{pmatrix}
			G_l  & wG_l & w^2G_l \\
			wG_l  & w^2G_l &  G_l 
		\end{pmatrix}.
	\end{equation}
%	
%	The ASEP code $C_{A_\frac{k_2}{2}}$ can be written as three equivalent parts, one of which can be expressed as
%	\begin{equation}
%		A_\frac{k_2}{2}=\left( 
%		\begin{array}{c|c|c}
%			G_{\frac{k_2}{2}}  & w*G_{\frac{k_2}{2}} & w^2*G_{\frac{k_2}{2}} \\
%			w*G_{\frac{k_2}{2}}  & w^2*G_{\frac{k_2}{2}} &  G_{\frac{k_2}{2}} 
%		\end{array}
%		\right),
%	\end{equation}
%	where $G_{\frac{k_2}{2}}$ is a generator matrix of quaternary $\frac{k_2}{2}$-dimensional Simplex code.
\end{lemma}
\begin{proof}
	This lemma follows from the properties of Simplex codes.
\end{proof}
% In addition, higher integer dimensions ASEP codes can be generated similarly. 

It is easy to check that, if $k\ge4$ is even, then the codes in Lemma \ref{divided_Simplex} and (1) of Lemma \ref{APS_construction} have the same parameters, where $l=\frac{k}{2}$. 
For convenience, in the paper, we agree that the parameters of the ASEP codes are all $[2^{k}-1,\frac{k}{2} ,3\cdot2^{k-2}]_4$, and denote them as $C_{A_{\frac{k}{2}}}$, and default its generator matrix as $A_{\frac{k}{2}}$.
%Here, we give an alternative way of iteratively construct ASEP codes.

%Here is a crucial derivation for constructing high-dimensional ASEP codes, especially for non-integer dimensions.

\begin{lemma}\label{aug_Simplex}
	Let the symbols be the same as above. Then, the following results hold.
	
	(1) Verticaljoin $\bm{0_{2^{k}-1}}$, $\bm{1_{2^{k}-1}}$ and  $\bm{w_{2^{k}-1}}$ with the generator matrix of $C_{A_{\frac{k}{2}}}$ yields an additive Griesmer code $C_{Au_\frac{k}{2}}$ with parameters $[2^{k}-1,\frac{k}{2}+1 ,3\cdot2^{k-2}-1]_4$, and there is 
	\begin{equation}
		[2^k-1,1,2^k-1]_4\subset [2^k-1,\frac{k}{2}+1,3\cdot 2^{k-2}-1]_4;
	\end{equation}

	(2) $C_{Au_\frac{k}{2}}$ is a three-weight code with weight distribution $1+(2^{k+2}-2^k-3)\cdot z^{3\cdot 2^{k-2}-1}+(2^k-1)\cdot z^{3\cdot 2^{k-2}}+3z^{2^k-1}$; 
	
	(3) Extending $C_{Au_\frac{k}{2}}$ gives an additive both GPO and Griesmer code $C_{Aue_\frac{k}{2}}$ with parameters $[2^k,\frac{k}{2}+1,3\cdot 2^{k-2}]_4$, there also exists
	\begin{equation}
		[2^k,1,2^k]_4\subset [2^k,\frac{k}{2}+1,3\cdot 2^{k-2}]_4;
	\end{equation}
	
	(4) $C_{Aue_\frac{k}{2}}$ is a two-weight code with weight distribution $1+(2^{k+2}-4)\cdot z^{3\cdot 2^{k-2}}+3z^{2^k}$. 
	
\end{lemma}
\begin{proof}
	For (1), since $C_{A_\frac{k}{2}}$ is an ASEP code, adding $\bm{0_{2^{k}-1}}$ to $C_{A_\frac{k}{2}}$ makes no difference, and for given any codeword $\bm{c}$ of $C_{A_\frac{k}{2}}$ and nonzero element $\beta$ in $F_4$, we have $\mathbf{w}_0(\bm{c}+\bm{\beta_{2^{k}-1}})=2^{k-2}$ and 
	$
		S=\{ \mathbf{w}_1(\bm{c}+\bm{\beta_{2^{k}-1}}), \mathbf{w}_w(\bm{c}+\bm{\beta_{2^{k}-1}}),\mathbf{w}_{w^2}(\bm{c}+\bm{\beta_{2^{k}-1}})\}=\{2^{k-2}-1,2^{k-2} \}.
	$
	 Furthermore, there is only one value of $2^k-1$ in $\mathbf{w}_1(\bm{c}+\bm{\beta_{2^{k}-1}})$, $\mathbf{w}_w(\bm{c}+\bm{\beta_{2^{k}-1}})$ and $\mathbf{w}_{w^2}(\bm{c}+\bm{\beta_{2^{k}-1}})$. Thus, $d_H(C_{Au_\frac{k}{2}})\ge 3\cdot2^{k-2}-1$. Any basis of $C_{A_\frac{k}{2}}$ can not span $\bm{1_{2^{k}-1}}$, $\bm{w_{2^{k}-1}}$ and  $\bm{w^2_{2^{k}-1}}$, so the dimension of $C_{Au_\frac{k}{2}}$ is $\frac{k}{2}+1$. Since $g(k+2,3\cdot2^{k-1}-2)=3(2^{k}-1)$, $[2^{k}-1,\frac{k}{2}+1 ,3\cdot2^{k-2}-1]_4$ code is an additive Griesmer code.
	
	For (2),  since $C_{Au_\frac{k}{2}}$ has two disjoint subcodes $[2^k-1,1,2^k-1]_4$ and $[2^{k}-1,\frac{k}{2} ,3\cdot2^{k-2}]_4$, respectively, all the codewords of weight $3\cdot2^{k-2}-1$ are a combination of codewords of these two subcodes. Therefore, $C_{Au_\frac{k}{2}}$ has weight distribution $1+(2^{k+2}-2^k-3)\cdot z^{3\cdot 2^{k-2}-1}+(2^k-1)\cdot z^{3\cdot 2^{k-2}}+3z^{2^k-1}$.  
	
	For (3), for the reason that there is $[2^{k}-1,\frac{k}{2},3\cdot2^{k-2}]_4\subset [2^k-1,\frac{k}{2}+1,3\cdot 2^{k-2}-1]_4$, by extending additive code $[2^k-1,1,2^k-1]_4$ we can get a $[2^k,\frac{k}{2}+1,3\cdot 2^{k-2}]_4$ code.
	Since $g(k+2,3\cdot2^{k-1})=3\cdot2^k$ and $g(k+2,3\cdot2^{k-1}+2)=3\cdot2^k+k+1>3\cdot2^k+3$, $[2^k,\frac{k}{2}+1,3\cdot 2^{k-2}]_4$ code is an additive both GPO and Griesmer code.
	
	For (4), since all combinations of the codewords in $[2^k,1,2^k]_4$ code and $[2^{k}-1,\frac{k}{2},3\cdot2^{k-2}]_4$ code give codewords of weight $3\cdot2^{k-2}$, $C_{Aue_\frac{k}{2}}$ is a two-weight code with weight distribution  $1+(2^{k+2}-4)\cdot z^{3\cdot 2^{k-2}}+3z^{2^k}$.
\end{proof}

Based on Lemma \ref{aug_Simplex}, we give the following method to construct high-dimensional ASEP codes from low-dimensional ASEP codes iteratively.
\begin{lemma}\label{iterating}
	Let the symbols be the same as above.
	Let $C_{A_{\frac{k}{2}}}$ be a $\frac{k}{2}$-dimensional ASEP code. Then the ($\frac{k}{2}+1$)-dimensional ASEP code $C_{A_{\frac{k}{2}+1}}$ can be generated by 
	\begin{equation}\label{E_iterating}
			\setlength{\arraycolsep}{0pt}
		A_{\frac{k}{2}+1}=\begin{pmatrix}
			A_{\frac{k}{2}}      & A_{\frac{k}{2}}      & A_{\frac{k}{2}}        & A_{\frac{k}{2}}        & \bm{0_{k\times 1}} & \bm{0_{k\times 1}} & \bm{0_{k\times 1}} \\
			\bm{0_{1,2^k-1}} & \bm{1_{1,2^k-1}} & \bm{w_{1,2^k-1}}   & \bm{w^2_{1,2^k-1}} & 1                      & w                      & w^2                    \\
			\bm{0_{1,2^k-1}} & \bm{w_{1,2^k-1}} & \bm{w^2_{1,2^k-1}} & \bm{1_{1,2^k-1}}   & w                      & w^2                    & 1
		\end{pmatrix}.
	\end{equation}
\end{lemma}
\begin{proof}
	The matrix $A_{\frac{k}{2}+1}$ can be divided into four parts as follows:
	\begin{equation}
					\setlength{\arraycolsep}{0pt}
		A_{\frac{k}{2}+1}=\left( 
		\begin{array}{c|cc|cc|cc}
			A_{\frac{k}{2}}      & A_{\frac{k}{2}}      & \bm{0_{k\times 1}} & A_{\frac{k}{2}}        & \bm{0_{k\times 1}} & A_{\frac{k}{2}}        & \bm{0_{k\times 1}} \\
			\bm{0_{1,2^k-1}} & \bm{1_{1,2^k-1}} & 1                      & \bm{w_{1,2^k-1}}   & w                      & \bm{w^2_{1,2^k-1}} & w^2                    \\
			\bm{0_{1,2^k-1}} & \bm{w_{1,2^k-1}} & w                      & \bm{w^2_{1,2^k-1}} & w^2                    & \bm{1_{1,2^k-1}}   & 1
		\end{array}
		\right).
	\end{equation}

	By Lemma \ref{aug_Simplex}, the code $C_{Aue_\frac{k}{2}}$ generated by 
	\begin{equation} \label{Eq_aug_Simplex}
		\begin{pmatrix}
			A_{\frac{k}{2}}	& \bm{0_{k, 1}}\\
			\bm{\beta_{2^{k}-1}}&\beta	\\
			w\cdot \bm{\beta_{2^{k}-1}}&w\cdot \beta
		\end{pmatrix}
	\end{equation} is a two-weight additive code with parameters $[2^{k},\frac{k}{2}+1 ,3\cdot2^{k-2}]_4$.

				\setlength{\arraycolsep}{0pt}
	$A^{\prime }_{\frac{k}{2}+1}=\left( 
	\begin{array}{cc|cc|cc}
		A_{\frac{k}{2}} & \bm{0_{k, 1}} & A_{\frac{k}{2}} & \bm{0_{k, 1}} & A_{\frac{k}{2}} & \bm{0_{k, 1}} \\
		\bm{1_{1,2^k-1}}     & 1             & \bm{w_{1,2^k-1}}     & w             & \bm{w^2_{1,2^k-1}}   & w^2           \\
		\bm{w_{1,2^k-1}}     & w             & \bm{w^2_{1,2^k-1}}   & w^2           & \bm{1_{1,2^k-1}}     & 1
	\end{array}
	\right)$ generates a two-weight additive code $C_{A_{\frac{k}{2}+1}}^{\prime}$ with parameters $[3\cdot 2^k,\frac{k}{2}+1,9\cdot 2^{k-2}]_4$. For any nonzero codeword $\bm{c^{\prime}}$ in $C_{A_{\frac{k}{2}+1}}^{\prime}$, there is $\mathbf{w}_0(\bm{c^{\prime}})=\mathbf{w}_1(\bm{c^{\prime}})=\mathbf{w}_w(\bm{c^{\prime}})=\mathbf{w}_{w^2}(\bm{c^{\prime}})=3\cdot 2^{k-2}$ or $\mathbf{w}_1(\bm{c^{\prime}})=\mathbf{w}_w(\bm{c^{\prime}})=\mathbf{w}_{w^2}(\bm{c^{\prime}})=2^{k}$. 
	Thus, $A_{\frac{k}{2}+1}$ can generate an additive code with parameters $[2^{k+2}-1,\frac{k}{2}+1,3\cdot 2^k]_4$, which satisfies for any nonzero codeword$\bm{c^{\prime\prime}}$ in $C_{A_{\frac{k}{2}+1}}$, $\mathbf{w}_1(\bm{c^{\prime\prime}})=\mathbf{w}_w(\bm{c^{\prime\prime}})=\mathbf{w}_{w^2}(\bm{c^{\prime\prime}})= 2^{k}$. 
	Therefore, it is indeed ($\frac{k}{2}+1$)-dimensional ASEP code $C_{A_{\frac{k}{2}+1}}$.
\end{proof}

By Lemma \ref{iterating}, we also can iteratively construct $C_{A_{\frac{k}{2}}}$ from $C_{1.5}$ for $k>3$ and $k$ is odd, such that it has an ordered matrix structure.

\begin{lemma}\label{combination_X_construction}
	If there exists an additive code $C_a$ with parameters $[n,\frac{k}{2},d]_4$, then there also exists an additive code with parameters
	\begin{equation}
		\left[s\cdot 2^k+n,\frac{k}{2}+1, 3s\cdot 2^{k-2}+d \right]_4
	\end{equation}
	 satisfies $s\ge \lceil \frac{d}{2^{k-2}} \rceil$.
\end{lemma}
 {\begin{proof}
According to (3) of Lemma \ref{aug_Simplex}, we have an special additive $[2^k,\frac{k}{2}+1,3\cdot 2^{k-2}]_4$ code $C_{Aue_\frac{k}{2}}$, which has a $[2^k,1,2^k]_4$ subcode.  
Therefore, the juxtaposition of $s$ generator matrix of $C_{Aue_\frac{k}{2}}$ produces additive $[s\cdot 2^k,\frac{k}{2}+1,3s\cdot 2^{k-2}]_4$ code with an $[s\cdot 2^k,1,s\cdot2^k]_4$ subcode.
Furthermore, since $s\ge \lceil \frac{d}{2^{k-2}} \rceil$, using $C_a$ as auxiliary code, we can get an additive $[s\cdot 2^k+n,\frac{k}{2}+1, 3s\cdot 2^{k-2}+d]_4$ code by Lemma \ref{construction_X}.
\end{proof}}

Lemma \ref{combination_X_construction} is also an effective method for constructing high-dimensional additive codes from low-dimensional additive codes. Lemma \ref{iterating} is a special case of Lemma \ref{combination_X_construction}. 

%\footnote{Multiset: A set containing repeated elements. If $\mathcal{L}$ has repeated elements, we call it a generator multiset.}
\section{Additive generalized anticode construction}\label{Sec IV anti}

 {Let $C_a$ be an additive code with generator matrix $A$. We define \textit{generator multiset} $\mathcal{L}$ of $C_a$ as a set consisting of column vectors of $A$. %\footnote{$\Phi(\mathcal{L})^{-1}$ can also be viewed as multiset lines in $PG(k-1, 2)$.}
Define $d_{max}(C_a)$ as the maximum Hamming weight of $C_a$.
%Let $\mathcal{L}_1$ and $\mathcal{L}_2$ be two sets of lines in $PG(k-1,2)$ that satisfy $\mathcal{L}_2 \subset \mathcal{L}_1$, and $\mathcal{L}_1\setminus  \mathcal{L}_2$ is the set that results from deleting $\mathcal{L}_2$ from $\mathcal{L}_1$. 
For a $k \times n$ matrix $A$ with no zero column, we define $A^{\circ}=(A\mid \bm{0_{k}}^T)$. 
For any matrix with the type of $B=(A\mid \bm{0_{k}}^T)$, we define $(B)^{\ominus }=A$.}
Throughout this section, we let $k=k_1+k_2$ such that $k_1\geq 3$ and $k_2\geq 4$ are respective odd and even integers.
%assumed that $k_1\ge3$ and $k_2\ge4$ are odd and even, respectively, and $k=k_1+k_2$.
%Let $C_{A_{\frac{k}{2}}}^{\circ}$ (resp., ) denote codes (resp., matrices) with one all zero column, and $C_{A_{\frac{k}{2}}}^{\ominus }$ (resp., $A_{\frac{k}{2}}^{\ominus }$) denote codes (resp., matrices) that removes all-zero column.
%A high-dimensional ASEP code can also be formed by combining two low-dimensional ASEP codes.
%
%Higher-dimensional ASEP codes can also be constructed by combining two low-dimensional ASEP codes.
 {\begin{lemma} (Additive Generalized Anticode Construction)\label{anticode construction}
	Let $C_{a_1}$ and $C_{a_2}$ be two additive codes with parameters $[n_1,k,d_1]_4$ and $[n_2,\le k]_4$, respectively, where $n_2\le n_1$.  
	If their generator sets $\mathcal{L}_{a_1}$ and $\mathcal{L}_{a_2}$ satisfy $\mathcal{L}_{a_2} \subset \mathcal{L}_{a_1}$, then we call $C_{a_2}$ as an \textit{additive generalized anticode} of $C_{a_1}$. $\mathcal{L}_{a_1}\setminus  \mathcal{L}_{a_2}$ generates an additive $[n_1-n_2,\le k,\ge d_1-d_{max}(C_{a_2})]_4$ code. 
\end{lemma}
\begin{proof}
	Since $\mathcal{L}_{a_2} \subset \mathcal{L}_{a_1}$, we can reorder the columns of generator matrix of $C_{a_1}$ such that any codeword $\bm{c}=(\bm{c}_1,\bm{c}_2)$ of $C_{a_1}$ satisfy $\bm{c}_2 \in C_{a_1}$.
	Then, $\mathbf{w}(\bm{c}_1)\ge d_1-d_{max}(C_{a_2})$.  
	Puncture the coordinates of $\bm{c}_2$ gives a result code with minimum distance at least $d_1-d_{max}(C_{a_2})$. 
	In addition, the dimension of the result code does not exceed $k$.
\end{proof}}
 
% \begin{remark}
% 	The linear anticode construction mainly depends on $PG(k-1,2)$, by deleting the point set $Q$ from $PG(k-1,2)$, where $Q$ is called the anticode code. However, in the case of addition, this process will become complicated in $PG(k-1,2)$, so we give Lemma 8.
% \end{remark}

Next, we show the specific procedure for constructing optimal additive codes using ASEP codes by Lemma \ref{anticode construction}.

\begin{definition}
	Let $A=(\bm{\alpha_1},\bm{\alpha_2},\cdots,\bm{\alpha_{n_1}})$, $B=(\bm{\beta_1},\bm{\beta_2},\cdots,\bm{\beta_{n_2}})$, where $\bm{\alpha_i}$ and $\bm{\beta_j}$, $1\le i\le n_1$, $1\le j\le n_2$, are column vectors.
	Define 
	\begin{equation} 
					\setlength{\arraycolsep}{0.5pt}
		A \star   B=\left( 
		\begin{array}{cccccccccc}
			\bm{\alpha}_1 & \cdots & \bm{\alpha_{1}}  & \bm{\alpha_2} & \cdots & \bm{\alpha_{2}}  & \cdots & \bm{\alpha_{n_1}} & \cdots & \bm{\alpha_{n_1}} \\
			\bm{\beta_1}  & \cdots & \bm{\beta_{n_2}} & \bm{\beta_1}  & \cdots & \bm{\beta_{n_2}} & \cdots & \bm{\beta_1}      & \cdots & \bm{\beta_{n_2}}
		\end{array}
		\right).
	\end{equation}
\end{definition}

\begin{lemma}\label{L_camba}
	 Let $C_{A_{\frac{k_1}{2}}}$ and $C_{A_{\frac{k_2}{2}}}$ be ASEP codes with generator matrices $A_{\frac{k_1}{2}}$ and $A_{\frac{k_2}{2}}$, generated by Lemmas \ref{iterating} and \ref{divided_Simplex}, respectively. 
	 Then %adding the all-zero column yields $A_{\frac{k_1}{2}}^{\circ}$ and $A_{\frac{k_2}{2}}^{\circ}$.  Then, 
%		\begin{equation} \label{comba}
%		A_\frac{k}{2}=(A_{\frac{k_1}{2}}^{\circ} \star A_{\frac{k_2}{2}}^{\circ})^{\ominus}=\left(A_{\frac{k_1}{2}} \star   A_{\frac{k_2}{2}} \begin{array}{|cc|c}
%			A_{\frac{k_1}{2}}&\bm{0_{k_1, 2^{k_2}-1}}\\
%			\bm{0_{k_2, 2^{k_1}-1}}&A_{\frac{k_2}{2}}
%		\end{array}
%	\begin{matrix}
%		\bm{0_{k, 1}}
%	\end{matrix}
%	\right).
%	\end{equation}
		\begin{equation} \label{comba}
						\setlength{\arraycolsep}{0pt}
	A_\frac{k}{2}=(A_{\frac{k_1}{2}}^{\circ} \star A_{\frac{k_2}{2}}^{\circ})^{\ominus}
	\cong 
	 \left(A_{\frac{k_1}{2}} \star   A_{\frac{k_2}{2}} \begin{array}{|cc}
		A_{\frac{k_1}{2}}&\bm{0_{k_1, 2^{k_2}-1}}\\
		\bm{0_{k_2, 2^{k_1}-1}}&A_{\frac{k_2}{2}}
	\end{array}
	\right).
\end{equation}
generates $\frac{k}{2}$-dimensional ASEP code $C_{A_\frac{k}{2}}$.
\end{lemma}
\begin{proof}
	In the case of using $C_{A_{1.5}}$ as initial ASEP code, the iteratively constructed $A_{\frac{k}{2}}$ satisfies that
	  lower $k_2$-rows of $A_{\frac{k}{2}}$ in Equation (\ref{E_iterating}) is that all vectors in $F_{4_a}^{\frac{k_2}{2}}$ repeat $2^{k_1}$ times, which can be written as
		\begin{equation}\label{repeat_matrix} 
		A_\frac{k}{2}^{\circ}=		\left( \begin{array}{cccc}
			A_{\frac{k_1}{2}}^{\circ}&A_{\frac{k_1}{2}}^{\circ}&\ldots&A_{\frac{k_1}{2}}^{\circ}\\
			\bm{0_{k_2,2^{k_1}}}&\bm{\mathsf{v}^{(1)}_{k_2,2^{k_1}}}&\ldots&\bm{\mathsf{v}^{(2^{k_2}-1)}_{k_2,2^{k_1}}}
		\end{array}
		\right),
	\end{equation}
	where $\bm{\mathsf{v}^{(i)}_{k_2,2^{k_1}}}$, $1\le i\le 2^{k_2}-1$, denotes the $k_2\times 2^{k_1}$ matrix generated by repeating the $i$-th nonzero vector in $F_{4_a}^{\frac{k_2}{2}}$ $2^{k_1}$ times.
	
		In addition, ASEP codes with integer dimensions generated by Equation (\ref{integer_ASEP}) are indeed all nonzero vectors in in $F_{4_a}^{\frac{k_2}{2}}$. 
	Therefore, it is easy to verify that rearranging the columns in Equation (\ref{E_iterating}) yields Equation (\ref{comba}). 
\end{proof}

Let $C_\frac{k}{2}$ be the ASEP code determined by Equation (\ref{comba}), and denote its generator set as $\mathcal{L}_\frac{k}{2}$.
 Let 
 \begin{equation} \label{divided_A_k_anti}
 	A_\frac{k_1^{\prime}}{2}=\left( \begin{array}{c}
 		A_{\frac{k_1}{2}}\\
 		\bm{0_{k_2\times (2^{k_1}-1)}}
 	\end{array}\right),	
 	A_\frac{k_2^{\prime}}{2}=\left( \begin{array}{c}
 		\bm{0_{k_1\times (2^{k_2}-1)}}\\
 		A_{\frac{k_2}{2}}
 	\end{array}\right),
 \end{equation}where $A_{\frac{k_1}{2}}$ and $A_{\frac{k_2}{2}}$ generated by Lemmas \ref{iterating} and \ref{divided_Simplex}, respectively. 
$\mathcal{L}_{\frac{k_1^{\prime}}{2}}$ and $\mathcal{L}_{\frac{k_2^{\prime}}{2}}$ are generator sets of $A_\frac{k_1^{\prime}}{2}$ and $A_\frac{k_2^{\prime}}{2}$, respectively. 
Obviously, $\mathcal{L}_{\frac{k_1^{\prime}}{2}}$ and $\mathcal{L}_{\frac{k_2^{\prime}}{2}}$ are two disjoint subsets of $\mathcal{L}_\frac{k}{2}$.

% Since the ASEP codes $C_\frac{k}{2}$ determined by Lemma \ref{L_camba} have an order matrix structure, $\mathcal{L}_{\frac{k_1}{2}}$ and $\mathcal{L}_{\frac{k_2}{2}}$  subsets of $\mathcal{L}_\frac{k}{2}$.

Therefore, the additive codes $C_\frac{k_1^{\prime}}{2}$ and $C_\frac{k_2^{\prime}}{2}$ determined by $A_\frac{k_1^{\prime}}{2}$ and $A_\frac{k_2^{\prime}}{2}$ are generalized anticodes of $C_\frac{k}{2}$, and their maximum weight are separately $3\cdot 2^{k_1-2}$ and $3\cdot 2^{k_2-2}$.

% the additive anticode obtained by deleting columns of $(C_{A_{\frac{k_1}{2}}}, \cdots,C_{A_{\frac{k_t}{2}}})$ from $C_{A_\frac{k}{2}}$.

\begin{theorem}\label{anticode_k_1_k_2}
Let the symbols be the same as above. Then, we have the following results.

(1) $\mathcal{L}_\frac{k}{2}\setminus \mathcal{L}_\frac{k_1^{\prime}}{2}$ generates an additive two-weight code with parameters $[2^k-2^{k_1}, \frac{k}{2}, 3\cdot 2^{k-2}-3\cdot 2^{k_1-2}]_4$, which has the weight distribution  $1+(2^k-2^{k_2})\cdot z^{3\cdot 2^{k-2}-3\cdot 2^{k_1-2}}+(2^{k_2}-1)\cdot z^{3\cdot 2^{k-2}}$;

(2) If $k_1=3$, then $[2^k-2^{k_1}, \frac{k}{2}, 3\cdot 2^{k-2}-3\cdot 2^{k_1-2}]_4$ code is GDO;  if $k_1\ge5$, it is GPO. 

(3) $\mathcal{L}_\frac{k}{2}\setminus (\mathcal{L}_{\frac{k_1^{\prime}}{2}}+\mathcal{L}_{\frac{k_2^{\prime}}{2}})$ generates an additive GDO three-weight code with parameters $[2^k-2^{k_1}-2^{k_2}+1, \frac{k}{2}, 3\cdot 2^{k-2}-3\cdot 2^{k_1-2}-3\cdot 2^{k_2-2}]_4$, which has the weight distribution $1+(2^k-2^{k_1}-2^{k_2})\cdot z^{3\cdot 2^{k-2}-3\cdot 2^{k_1-2}-3\cdot 2^{k_2-2}}+(2^{k_1}-1)\cdot z^{3\cdot 2^{k-2}-3\cdot 2^{k_1-2}}+(2^{k_2}-1)\cdot z^{3\cdot 2^{k-2}-3\cdot 2^{k_2-2}}$.

(4) $[2^k-2^{k_1}-2^{k_2}+1, \frac{k}{2}, 3\cdot 2^{k-2}-3\cdot 2^{k_1-2}-3\cdot 2^{k_2-2}]_4$ code is GDO if $|k_1-k_2|=1$ and $9\le k\le15$, or $|k_1-k_2|\ge3$ and $\min\{k_1,k_2\}\le4$; else if $|k_1-k_2|=1$ and $k\ge17$, or $|k_1-k_2|\ge3$ and $\min\{k_1,k_2\}>4$, it is GPO.
\end{theorem}
\begin{proof} 
In accordance with the foregoing, the corresponding matrix $A_{\frac{k}{2}\setminus\frac{k_1^{\prime}}{2}}$ of	$\mathcal{L}_\frac{k}{2}\setminus \mathcal{L}_\frac{k_1^{\prime}}{2}$ can be written as the following form:
	
			\begin{equation} 
	A_{\frac{k}{2}\setminus\frac{k_1^{\prime}}{2}}=\left(A_{\frac{k_1}{2}} \star   A_{\frac{k_2}{2}} \begin{array}{|c}
			\bm{0_{k_1, 2^{k_2}-1}}\\
			A_{\frac{k_2}{2}}
		\end{array}
		\right).
	\end{equation}
	
	Denote the additive codes generated by the upper $k_1$ rows and lower $k_2$ rows of $A_{\frac{k}{2}\setminus\frac{k_1^{\prime}}{2}}$ as $C_{up}$ and $C_{down}$.
	Then, by the properties of $A_{\frac{k}{2}\setminus\frac{k_1^{\prime}}{2}}$, for all nonzero codewords $\bm{c_1}\in C_{up}$ and $\bm{c_2}\in C_{down}$, there are $\mathbf{w}(\bm{c_1 })=\mathbf{w}(\bm{c_1 }+\bm{c_2})=3\cdot 2^{k-2}-3\cdot 2^{k_1-2}$ and $\mathbf{w}(\bm{c_2})=3\cdot 2^{k-2}$. 
	Therefore, the intersection of $C_{up}$ and $C_{down}$ is a zero codeword, it is easy to conclude that $\mathcal{L}_\frac{k}{2}\setminus \mathcal{L}_\frac{k_1^{\prime}}{2}$ generates a two-weight $[2^k-2^{k_1}, \frac{k}{2}, 3\cdot 2^{k-2}-3\cdot 2^{k_1-2}]_4$ code with weight distribution  $1+(2^k-2^{k_2})\cdot z^{3\cdot 2^{k-2}-3\cdot 2^{k_1-2}}+(2^{k_2}-1)\cdot z^{3\cdot 2^{k-2}}$. 
	Since $g(k,3\cdot 2^{k-1}-3\cdot 2^{k_1-1}+2)=3\cdot 2^{k}-3\cdot 2^{k_1}+k_1>3(2^k-2^{k_1})$, this code is GDO for $k_1=3$; if $k_1>3$, it is GPO.

	Similarly, the corresponding matrix of $\mathcal{L}_\frac{k}{2}\setminus (\mathcal{L}_{\frac{k_1^{\prime}}{2}}+\mathcal{L}_{\frac{k_2^{\prime}}{2}})$ is $A_{\frac{k_1}{2}} \star   A_{\frac{k_2}{2}}$, which also can be divided to up and down two parts.
	Let $C_{up}^{\prime}$ and $C_{down}^{\prime}$ denote the additive codes generated by the upper $k_1$ rows and lower $k_2$ rows of $A_{\frac{k_1}{2}} \star   A_{\frac{k_2}{2}}$, respectively.
	 Then, for all nonzero codewords $\bm{c_1^{\prime}}\in C_{up}^{\prime }$ and $\bm{c_2^{\prime}}\in C_{down}^{\prime}$, there are $\mathbf{w}(\bm{c_1^{\prime }})=3\cdot 2^{k-2}-3\cdot 2^{k_1-2}$, $\mathbf{w}(\bm{c_2^{\prime }})=3\cdot 2^{k-2}-3\cdot 2^{k_2-2}$ and $\mathbf{w}(\bm{c_1^{\prime }}+\bm{c_2^{\prime}})=3\cdot 2^{k-2}-3\cdot 2^{k_1-2}-3\cdot 2^{k_2-2}$.
	Thus, the intersection of $C_{up}^{\prime }$ and $C_{down}^{\prime }$ is also only zero codeword, $\mathcal{L}_\frac{k}{2}\setminus (\mathcal{L}_{\frac{k_1^{\prime}}{2}}+\mathcal{L}_{\frac{k_2^{\prime}}{2}})$ generates a three-weight $[2^k-2^{k_1}-2^{k_2}+1, \frac{k}{2}, 3\cdot 2^{k-2}-3\cdot 2^{k_1-2}-3\cdot 2^{k_2-2}]_4$ code with weight distribution $1+(2^k-2^{k_1}-2^{k_2})\cdot z^{3\cdot 2^{k-2}-3\cdot 2^{k_1-2}-3\cdot 2^{k_2-2}}+(2^{k_1}-1)\cdot z^{3\cdot 2^{k-2}-3\cdot 2^{k_1-2}}+(2^{k_2}-1)\cdot z^{3\cdot 2^{k-2}-3\cdot 2^{k_2-2}}$.

	%$g(k,3\cdot 2^{k-1}-3\cdot 2^{k_1-1}-3\cdot 2^{k_2-1}+2)>3(2^k-2^{k_1}-2^{k_2}+1)$, this code is also GDO.
		
	Since $g(k,3\cdot 2^{k-1}-3\cdot 2^{k_1-1}-3\cdot 2^{k_2-1}+2)=\left\{\begin{matrix}
				3\cdot 2^{k}-3\cdot 2^{k_1}-3\cdot 2^{k_2}+\min\{k_1,k_2\},& |k_1-k_2|=1,\\
				3\cdot 2^{k}-3\cdot 2^{k_1}-3\cdot 2^{k_2}+\min\{k_1,k_2\}+2,&|k_1-k_2|\ge 3,
			\end{matrix}\right.$ the 
		$[2^k-2^{k_1}-2^{k_2}+1, \frac{k}{2}, 3\cdot 2^{k-2}-3\cdot 2^{k_1-2}-3\cdot 2^{k_2-2}]_4$ code is GDO if $|k_1-k_2|=1$ and $9\le k\le15$, or $|k_1-k_2|\ge3$ and $\min\{k_1,k_2\}\le4$; else if $|k_1-k_2|=1$ and $k\ge17$, or $|k_1-k_2|\ge3$ and $\min\{k_1,k_2\}>4$, it is GPO.
\end{proof}

As agreed upon, $A_{\frac{k_2}{2}}$ is generated from Lemma \ref{divided_Simplex}, which can be divided into three equivalent blocks, i.e., $A_\frac{k_2^{\prime}}{2}$ can be written as $
	A_\frac{k_2^{\prime}}{2}=\left( 
		\begin{array}{c|c|c}
			\bm{0_{k_1\times \frac{2^{k_2}-1}{3}}}&\bm{0_{k_1\times \frac{2^{k_2}-1}{3}}}&\bm{0_{k_1\times \frac{2^{k_2}-1}{3}}}\\
			G_\frac{k_2}{2}  & wG_\frac{k_2}{2} & w^2G_\frac{k_2}{2} \\
			wG_\frac{k_2}{2}  & w^2G_\frac{k_2}{2} &  G_\frac{k_2}{2}
		\end{array}\right)
$, where $G_\frac{k_2}{2}$ is a generator matrix of $\frac{k_2}{2}$-dimensional quaternary Simplex code. 
Let $A_{\frac{k_2^{\prime}}{2},m}$ denote the first $m$ blocks of $A_{\frac{k_2^{\prime}}{2}}$, where $1\le m \le 3$.
Then, we have the following corollary.

%Let $\mathcal{L}_{\frac{k_2}{2},m}$ be the generator set of the first $m$ blocks $A_{\frac{k_2}{2},m}$ of $A_\frac{k_2^{\prime}}{2}$.
%
%Since $A_{\frac{k_2}{2}}$ can be divided into three equivalent blocks, here we denote the first $m$ blocks of $A_{\frac{k_2}{2}}$ by $A_{\frac{k_2}{2}}^m$.

%We denote the codes determined by $\mathcal{L}_\frac{k}{2}\setminus (\mathcal{L}_{\frac{k_1}{2}}+\mathcal{L}_{\frac{k_2}{2}})$ and $\mathcal{L}_\frac{k}{2}\setminus \mathcal{L}_\frac{k_i}{2}^{\prime}$ with $C_{\mathcal{L}_\frac{k}{2}}\setminus (C_{\mathcal{L}_{\frac{k_1}{2}}}+ C_{\mathcal{L}_{\frac{k_2}{2}}})$ and $C_{\mathcal{L}_\frac{k}{2}}\setminus C_{\mathcal{L}_{\frac{k_i}{2}}}$, respectively, where $i=1$ or $2$.

\begin{corollary} \label{anticode_k_1_divided_k_2}
	%	Let $k=k_1+k_2$, $k_1,k_2\ge3$, $k$, $k_1$ is odd, $1\le m \le 3$. Then, there exist the following results.
	%Let $1\le m \le 3$, then there exist the following results.
	Let the symbols be the same as above. 
	Let $\mathcal{L}_{\frac{k_2^{\prime}}{2},m}$ be the generator set of $C_{\frac{k_2^{\prime}}{2},m}$. Then, the following results hold.

	(1) $\mathcal{L}_\frac{k}{2}\setminus \mathcal{L}_{\frac{k_2^{\prime}}{2},m}^{\prime}$ generates a two-weight code with parameters $[2^k-1-m\cdot \frac{2^{k_2}-1}{3}, \frac{k}{2}, 3\cdot 2^{k-2}-m\cdot 2^{k_2-2}]_4$, whose weight distribution is $1+(2^k-2^{k_1})\cdot z^{3\cdot 2^{k-2}-m\cdot 2^{k_2-2}}+(2^{k_1}-1)\cdot z^{3\cdot 2^{k-2}}$; 
	
	(2) The $[2^k-1-m\cdot \frac{2^{k_2}-1}{3}, \frac{k}{2}, 3\cdot 2^{k-2}-m\cdot 2^{k_2-2}]_4$ code is an additive both GPO and Griesmer code for $m=1$, others are just GPO;

%(3) $\mathcal{L}_\frac{k}{2}\setminus (\mathcal{L}_{\frac{k_1^{\prime}}{2}}+\mathcal{L}_{\frac{k_2^{\prime}}{2},m}^{\prime})$ generates a three-weight code with parameters $[2^k-2^{k_1}-m\cdot \frac{2^{k_2}-1}{3}, \frac{k}{2}, 3\cdot 2^{k-2}-3\cdot 2^{k_1-2}-m\cdot 2^{k_2-2}]_4$, and weight distribution $1+(2^k-2^{k_1}-2^{k_2}+1)\cdot z^{3\cdot 2^{k-2}-3\cdot 2^{k_1-2}-m\cdot 2^{k_2-2}}+(2^{k_1}-1)\cdot z^{3\cdot 2^{k-2}-3\cdot 2^{k_1-2}}+(2^{k_2}-1)\cdot z^{3\cdot 2^{k-2}-m\cdot 2^{k_2-2}}$; 
%
%(4)If $|k_1-k_2|\ne 1$ or $m<3$, then $[2^k-2^{k_1}-m\cdot \frac{2^{k_2}-1}{3}, \frac{k}{2}, 3\cdot 2^{k-2}-3\cdot 2^{k_1-2}-m\cdot 2^{k_2-2}]_4$ code is GDO; else if $\min\{k_1,k_2\}>\lfloor \frac{m}{2}\rfloor+3$, it is GPO.
\end{corollary}
\begin{proof}
One can easily conclude that the above constructions hold according to the method of proving Theorem \ref{anticode_k_1_k_2}. 
Moreover, since $g(k,3\cdot 2^{k-1}-m\cdot 2^{k_2-1})=3\cdot 2^{k}-m\cdot 2^{k_2}+\lceil \frac{m}{2} \rceil-3$ and $g(k,3\cdot 2^{k-1}-m\cdot 2^{k_2-1}+2)=3\cdot 2^{k}-m\cdot 2^{k_2}+k_2+\lceil \frac{m}{2} \rceil-2>3(2^k-m\cdot \frac{2^{k_2}-1}{3})$, if $m=1$, $[2^k-1-m\cdot \frac{2^{k_2}-1}{3}, \frac{k}{2}, 3\cdot 2^{k-2}-m\cdot 2^{k_2-2}]_4$ code is both additive GPO and Griesmer code, otherwise, it is just GPO code.
%In addition, we have
%$$g(k,3\cdot 2^{k-1}-3\cdot 2^{k_1-1}-m\cdot 2^{k_2-1}+2)=\left\{\begin{matrix}
%	3\cdot 2^{k}-3\cdot 2^{k_1}- 2^{k_2}+k_2+1,                                                                                 & k_1-k_2=1, m=1,   \\
%	3\cdot 2^{k}-3\cdot 2^{k_1}- 2\cdot2^{k_2}+k_2-1,                                                                           & k_1-k_2=1, m=2,   \\
%	3\cdot 2^{k}-3\cdot 2^{k_1}- 2^{k_2}+k_1,                                                                                   & k_2-k_1=1, m=1,   \\
%	3\cdot 2^{k}-3\cdot 2^{k_1}- 2\cdot2^{k_2}+k_1+1,                                                                           & k_2-k_1=1, m=2,   \\
%	3\cdot 2^{k}-3\cdot 2^{k_1}- 3\cdot2^{k_2}+\min\{k_1,k_2\},                                                                             & |k_1-k_2|=1, m=3,   \\
%	3\cdot 2^{k}-3\cdot 2^{k_1}-m\cdot 2^{k_2}+\min\{k_1,k_2\}+\lceil \frac{m}{2} \rceil,                                       & |k_1-k_2|\ge 3,
%\end{matrix}\right.$$
%
%The case of $m=3$ has been proved in (4) of Theorem X, the code is GDO for $m\ne3$, or $m=3$ and $9\le k\le11$
%
%
%
% and GPO if $k\ge13$.
%
%$m=1$, $\min\{k_1,k_2\}>3$
%
%$m=2$, if $$\left\{\begin{matrix}
%	if k_1=k_2+1, k_2\ge8\\
%	if k_2=k_1+1, k_1\ge 5\\
%	if |k_1-k_2|\ge 3, \min\{k_1,k_2\}>4
%	
%\end{matrix}\right.$$
%$\min\{k_1,k_2\}>3$
%the code is GPO
%
%
%so if $|k_1-k_2|\ge 3$ or $m<3$, $[2^k-2^{k_1}-m\cdot \frac{2^{k_2}-1}{3}, \frac{k}{2}, 3\cdot 2^{k-2}-3\cdot 2^{k_1-2}-m\cdot 2^{k_2-2}]_4$ code is GDO;
%if $\min\{k_1,k_2\}>\lfloor \frac{m}{2}\rfloor+3$, then it is GPO. 
\end{proof}

%Based on Lemma \ref{divided_Simplex}, we prove that the following two additive codes can be divided into three equivalent parts.

%matrices is $A_{\frac{k_1}{2}}^{\circ} \star \begin{pmatrix}
%	G_{\frac{k_2}{2}}\\ w*G_{\frac{k_2}{2}}
%\end{pmatrix}$ Let the symbols be the same as above. Then
%
%(1) the additive codes generated by $\mathcal{L}_\frac{k}{2}\setminus \mathcal{L}_\frac{k_1^{\prime}}{2}$ and $\mathcal{L}_\frac{k}{2}\setminus (\mathcal{L}_{\frac{k_1^{\prime}}{2}}+\mathcal{L}_{\frac{k_2^{\prime}}{2}})$ both can be divided into three equivalent parts. 

\begin{theorem}\label{anticode_k_1_k_2_1/3}
	Let the symbols be the same as above, and then we will have the following results.
	
	(1) $\mathcal{L}_\frac{k}{2}\setminus \mathcal{L}_\frac{k_1^{\prime}}{2}$ can be divided into three equivalent parts, each part 
	generates an additive two-weight additive code with parameters $[\frac{2^k-2^{k_1}}{3}, \frac{k}{2}, 2^{k-2}-2^{k_1-2}]_4$, whose weight distribution is $1+(2^k-2^{k_2})\cdot z^{ 2^{k-2}- 2^{k_1-2}}+(2^{k_2}-1)\cdot z^{ 2^{k-2}}$;
	
	(2) $\mathcal{L}_\frac{k}{2}\setminus (\mathcal{L}_{\frac{k_1^{\prime}}{2}}+\mathcal{L}_{\frac{k_2^{\prime}}{2}})$ can be divided into three equivalent parts, each part generates a three-weight additive code with parameters $[\frac{2^k-2^{k_1}-2^{k_2}+1}{3}, \frac{k}{2}, 2^{k-2}-2^{k_1-2}-2^{k_2-2}]_4$, and weight distribution $1+(2^k-2^{k_1}-2^{k_2})\cdot z^{ 2^{k-2}- 2^{k_1-2}- 2^{k_2-2}}+(2^{k_1}-1)\cdot z^{ 2^{k-2}- 2^{k_1-2}}+(2^{k_2}-1)\cdot z^{ 2^{k-2}- 2^{k_2-2}}$.
	
	(3) The codes in (1) and (2) are both GPO and Griesmer codes.
\end{theorem}

\begin{proof}
	Since $\frac{k}{2}$-dimensional ASEP code $C_{A_\frac{k}{2}}$ can be generated by $(A_{\frac{k_1}{2}}^{\circ} \star A_{\frac{k_2}{2}}^{\circ})^{\ominus  }$, and $A_{\frac{k_2}{2}}$ can be divided into three parts, the generator matrix of $\mathcal{L}_\frac{k}{2}\setminus \mathcal{L}_\frac{k_1^{\prime}}{2}$ can be written as

	\begin{equation}
					\setlength{\arraycolsep}{0pt}
	\left( 
	\begin{array}{c|c|c}
		A_{\frac{k_1}{2}}^{\circ} \star \begin{pmatrix}
			G_{\frac{k_2}{2}}\\ w*G_{\frac{k_2}{2}}
		\end{pmatrix}  
	&A_{\frac{k_1}{2}}^{\circ} \star \begin{pmatrix}
		w*G_{\frac{k_2}{2}}\\ w^2*G_{\frac{k_2}{2}}
	\end{pmatrix}  
	&A_{\frac{k_1}{2}}^{\circ} \star \begin{pmatrix}
		w^2*G_{\frac{k_2}{2}}\\ G_{\frac{k_2}{2}}
	\end{pmatrix} 
	\end{array}
	\right).
\end{equation}

	Since the additive code generated by  $\begin{pmatrix}
	w^2*G_{\frac{k_2}{2}}\\ G_{\frac{k_2}{2}}
\end{pmatrix}$ is also a linear code, for all $\alpha$ in $F_4^*$, the additive codes generated by $\begin{pmatrix}
	\alpha\cdot G_{\frac{k_2}{2}}\\\alpha\cdot w*G_{\frac{k_2}{2}}
\end{pmatrix}$ are same.
Therefore, the codes generated by $A_{\frac{k_1}{2}}^{\circ} \star \begin{pmatrix}
	G_{\frac{k_2}{2}}\\ w*G_{\frac{k_2}{2}}
\end{pmatrix}$ and $A_{\frac{k_1}{2}}^{\circ} \star \begin{pmatrix}
	\alpha\cdot G_{\frac{k_2}{2}}\\\alpha\cdot w*G_{\frac{k_2}{2}}
\end{pmatrix}$ are the same.
Similarly, the codes generated by $A_{\frac{k_1}{2}} \star \begin{pmatrix}
	G_{\frac{k_2}{2}}\\ w*G_{\frac{k_2}{2}}
\end{pmatrix}$ and $A_{\frac{k_1}{2}} \star \begin{pmatrix}
	\alpha\cdot G_{\frac{k_2}{2}}\\\alpha\cdot w*G_{\frac{k_2}{2}}
\end{pmatrix}$ are also the same.

Thus, the additive codes generated by $\mathcal{L}_\frac{k}{2}\setminus \mathcal{L}_\frac{k_1^{\prime}}{2}$ and $\mathcal{L}_\frac{k}{2}\setminus (\mathcal{L}_{\frac{k_1^{\prime}}{2}}+\mathcal{L}_{\frac{k_2^{\prime}}{2}})$ both can be divided into three equivalent parts. 
The result additive codes have similar parameters and weight distributions with those in Theorem \ref{anticode_k_1_k_2}, but only need to replace the lengths and distances by one-third of the original codes. 

Since $g(k,2^{k-1}-2^{k_1-1})=2^{k}-2^{k_1}$ and $g(k,2^{k-1}-2^{k_1-1}+2)=2^{k}-2^{k_1}+k_1+1$, $[\frac{2^k-2^{k_1}}{3}, \frac{k}{2}, 2^{k-2}-2^{k_1-2}]_4$ code is both GPO and Griesmer code.
In addition, as $g(k,2^{k-1}-2^{k_1-1}-2^{k_2-1})=2^{k}-2^{k_1}-2^{k_2}+1$ and $g(k,2^{k-1}-2^{k_1-1}-2^{k_2-1}+2)=2^{k}-2^{k_1}-2^{k_2}+\min\{k_1,k_2\}+2>2^{k}-2^{k_1}-2^{k_2}+4$,  
$[\frac{2^k-2^{k_1}}{3}, \frac{k}{2}, 2^{k-2}-2^{k_1-2}]_4$ code is also both GPO and Griesmer code.
\end{proof}

%\begin{theorem}\label{anticode_k_1_k_2_1/3}
%Let the symbols be the same as above. Then
%
%(1) $A_{\frac{k_1}{2}}^{\circ} \star \begin{pmatrix}
%	G_{\frac{k_2}{2}}\\ w*G_{\frac{k_2}{2}}
%\end{pmatrix}$ generates a two-weight code with parameters $[\frac{2^k-2^{k_1}}{3}, \frac{k}{2}, 2^{k-2}-2^{k_1-2}]_4$, whose weight distribution is $1+(2^k-2^{k_2})\cdot z^{ 2^{k-2}- 2^{k_1-2}}+(2^{k_2}-1)\cdot z^{ 2^{k-2}}$;
%
%(2) $A_{\frac{k_1}{2}} \star \begin{pmatrix}
%	G_{\frac{k_2}{2}}\\ w*G_{\frac{k_2}{2}}
%\end{pmatrix}$ generates a three-weight code with parameters $[\frac{2^k-2^{k_1}-2^{k_2}+1}{3}, \frac{k}{2}, 2^{k-2}-2^{k_1-2}-2^{k_2-2}]_4$, and weight distribution $1+(2^k-2^{k_1}-2^{k_2})\cdot z^{ 2^{k-2}- 2^{k_1-2}- 2^{k_2-2}}+(2^{k_1}-1)\cdot z^{ 2^{k-2}- 2^{k_1-2}}+(2^{k_2}-1)\cdot z^{ 2^{k-2}- 2^{k_2-2}}$.
%\end{theorem}
%\begin{proof}
%In Lemma \ref{L_divided_camba}, we show that both additive codes in Theorem \ref{anticode_k_1_k_2} can be divided into three equivalent parts, so this theorem naturally holds.
%\end{proof}

\section{Generalized construction X}\label{Sec V ASEP}
%	Let	$\begin{matrix}
%	[n,k_1-k_2, \delta_1]_4\\
%	[n,k_2,\delta_2]_4
%\end{matrix}\subset [n,k_1,\delta_1]_4$ denote $[n,k_1,\delta_1]_4$ have two disjoint subcodes $[n,k_1-k_2, \delta_1]_4$ and $[n,k_2,\delta_2]_4$. 

To describe the matrix combination method in this section, we denote by $A_{[n,k]_4}$ or $A_{[n,k,d]_4}$ a generator matrix of the additive code with parameters $[n,k,d]_4$. 
%Let $C_{[n,k,d]_4}$ denote additive code with parameters $[n,k,d]_4$, whose generator matrix denoted as $A_{[n,k]_4}$ or $A_{[n,k,d]_4}$.
\begin{definition}
	If $C_a$ is an additive $[n,k_1,\delta_1]_4$ code which has two disjoint subcodes $C_{s_1}$ and $C_{s_2}$ with parameters $[n,k_1-k_2, \delta_1]_4$ and $[n,k_2,\delta_2]_4$, respectively, satisfying for all nonzero codewords $\bm{c_1}$ of $C_{s_1}$ and $\bm{c_2}$ of $C_{s_2}$, $\delta_2=\mathbf{w}(\bm{c_1}+\bm{c_2})> \delta_1$. Then, we call $C_a$ has an \textit{invariant subcode} $C_{s_2}$ with parameters $[n,k_2,\delta_2]_4$.
%	 if $\begin{matrix}
%		[n,k_1-k_2, \delta_1]_4\\
%		[n,k_2,\delta_2]_4
%	\end{matrix}\subset [n,k_1,\delta_1]_4$
\end{definition}

\begin{theorem}(Generalized Construction X)\label{Generalized_X_Construction}
	Let $C_1$ and $C_2$ be additive codes with parameters $[n_1,k,d_1]_4$ and $[n_2,k,d_2]_4$, respectively, satisfying $[n_1,k_1^{\prime},d_1^{\prime}]_4\subset [n_1,k,d_1]_4$, and $[n_2,k_2^{\prime},d_2^{\prime}]_4\subset [n_2,k,d_2]_4$. 
	Suppose $k_1=k-k_1^{\prime}$, $k_2=k-k_2^{\prime}$, $k_1>k_2$,  $\delta_1=d_1^{\prime}-d_1$, $\delta_2=d_2^{\prime}-d_1+\delta_1$. Choosing auxiliary code $C_{au}$ with parameters $[n,k_1,\delta_1]_4$,  which has an invariant subcode $C_{isub}$ with parameters $[n,k_2,\delta_2]_4$, 
	then there is an additive code with parameters $[n_1+n_2+n,k,d_1+d_2+\delta_2]_4$.
\end{theorem}
\begin{proof}
	Firstly, $C_1$ and $C_{au}$ can determine an additive code $C_{X_1}$ with parameters $[n_1+n,k,d_1+\delta_1]_4$ by Construction X, whose generator matrix can be written as $A_{X_1}$.
	\begin{equation}\label{First_step_X}
		\begin{array}{rl}
				A_{X_1}=&\left(	\begin{array}{c|c}
					A_{[n_1,k_1, d_1]_4} & A_{[n,k_1,\delta_1]_4}\\
					A_{[n_1,k_1^{\prime},d_1^{\prime}]_4} & \bm{0_{(k-k_1), n}} \\  
				\end{array}\right)	\\
			=&\left(	\begin{array}{c|c}
				A_{[n_1,k_1,  d_1]_4} & \begin{matrix}
					A_{[n,k_2,\delta_2]_4}\\
					A_{[n,k_1-k_2, \delta_1]_4}
				\end{matrix}\\
				A_{[n_1,k_1^{\prime},d_1^{\prime}]_4} & \bm{0_{k_1^{\prime}, n}}\\  
			\end{array}\right).
		\end{array}
	\end{equation}
	
	Let $C_{X_{1a}}$ and $C_{X_{1b}}$ be the additive codes determined by the first $k_2$ rows of $A_{X_1}$ and the others rows, respectively. 
Since $C_{au}$ has an invariant subcode $C_{isub}$ with parameters $[n,k_2,\delta_2]_4$ satisfying for any nonzero codeword $\bm{c}$ in $C_{au} \setminus C_{isub}$, $\mathbf{w}(\bm{c})\ge \delta_2$,
then for all nonzero codewords $\bm{c_1}$ of $C_{X_{1a}}$ and $\bm{c_2}$ of $C_{X_{1b}}$,  we have $\mathbf{w}(\bm{c_1}+\bm{c_2})\ge d_1+\delta_2$.  
Therefore, $C_{X_{1a}}$ can help to increase the distance of a $\frac{k_2}{2}$-dimensional subspace of $C_2$ by $\delta_2-\delta_1$. The specific combination method is as following matrix.

\begin{equation}
				\setlength{\arraycolsep}{0.5pt}
	A_{X_2}
	=\left(	\begin{array}{c|c|c}
		A_{[n_1,k_1, d_1]_4} &A_{[n_2,k_2,  d_2]_4} & \begin{matrix}
			A_{[n,k_2,\delta_2]_4}\\
			A_{[n,k_1-k_2, \delta_1]_4}
		\end{matrix}\\
		A_{[n_1,k_1^{\prime},d_1^{\prime}]_4}& A_{[n_2,k_2^{\prime},d_2^{\prime}]_4}& \bm{0_{k_1^{\prime}, n}}\\  
	\end{array}\right).
\end{equation}

It is easy to judge that $A_{X_2}$ can generate an additive code with parameters $[n_1+n_2+n,k,d_1+d_2+\delta_2]_4$. 
\end{proof}

Theorem \ref{Generalized_X_Construction} is an enhancement of Construction X. Compared with Construction X, the present construction enlarges the distance by $\delta_2$ rather than $\delta_1$. For convenience, we also have the following corollary.

\begin{corollary}\label{cor_Generalized_X_Construction}
Let $C_1$ and $C_2$ be additive codes with parameters $[n_1,k,d_1]_4$ and $[n_2,k,d_2]_4$, respectively. 
If $C_1$ has a subcode with parameters $[n_1,k_1^{\prime},d_1^{\prime}]_4$ and $C_2$ has an invariant subcode with parameters $[n,k-k_1^{\prime},d_2+d_1^{\prime}-d_1]_4$, then there exists an additive code with parameters $[n_1+n_2,k,d_1^{\prime}+d_2]_4$. 
\end{corollary}
\begin{proof}
%This corollary holds from proving Theorem \ref{Generalized_X_Construction}.
This corollary holds from the procedure of proving Theorem \ref{Generalized_X_Construction}.
\end{proof}

Additive codes with invariant subcodes are essential in implementing the generalized Construction X. In the following, we focus on constructing optimal additive codes with invariant subcodes.

\begin{lemma}\label{One-third}
	Let $k\ge 5$ be an odd number. Then, there exists an additive GPO $[ \frac{2^k+1}{3} , \frac{k}{2},2^{k-2}]_4$ code with weight distribution $1+(2^{k-1}-1)z^{2^{k-2}}+2^{k-1}z^{2^{k-2}+1}$, which has an invariant subcode with parameters $[\frac{2^k+1}{3} , 0.5,2^{k-2}+1]_4$. 
\end{lemma}
\begin{proof}
	Here, we use an iterative method to prove the existence of $[ \frac{2^k+1}{3} , \frac{k}{2},2^{k-2}]_4$ code.
	For $k=3$, we have an additive $[3,1.5,2]_4$ code with generator matrix $A_{[3,1.5,2]_4}$ as follows.
	\begin{equation}\label{(3,1.5,2)}
		A_{[3,1.5,2]_4}=\begin{pmatrix}
			w &  w &w\\
			1  & 1 &  0\\
			0  & 1& 1
		\end{pmatrix}.
	\end{equation}
	
	This is a two-weight code with weight distribution $1+3z^2+4z^3$, which has an invariant subcode with parameters $[3, 0.5,3]_4$. 
	
	By Lemma \ref{aug_Simplex}, we can get a two-weight additive $[2^k,\frac{k}{2}+1,3\cdot 2^{k-2}]_4$ code, which has a $[2^k,1,2^k]_4$ subcode. 
	Let $k=5$, it is $[8,2.5,6]_4$ code, and has an $[8,1,8]_4$ subcode. Choosing $[3,1.5,2]_2$ as an auxiliary code, then by Lemma \ref{construction_X}, we can get an $[11,2.5,8]_4$. According to Lemma \ref{aug_Simplex}, a generator matrix of $[8,2.5,6]_4$ code can be denoted as $A_{[8,2.5,6]_4}=\begin{pmatrix}
	\bm{\gamma_1}^T,\bm{\gamma_2}^T,\bm{\gamma_3}^T,\bm{1_{1,2^{k}-1}}^T,\bm{w_{1,2^{k}-1}}^T
\end{pmatrix}^T$, where $\bm{\gamma_i}$ are row vectors, $1\le i\le 3$. 
Thus, the generator matrix of $[11,2.5,8]_4$ can be written as 
\begin{equation}
	A_{[11,2.5,8]_4}=\left( 
	\begin{array}{c|c}
		\bm{\gamma_1}  & w, w,w \\
		\bm{\gamma_2}  & 1, 1,  0\\ 
		\bm{\gamma_3}  & 0, 1, 1\\ 
		\bm{1_{1,2^{k}-1}}& \bm{0_{1,3}}\\ 
		\bm{w_{1,2^{k}-1}}& \bm{0_{1,3}}\\ 
	\end{array}
	\right)
	=\left( 
	\begin{array}{c}
		\bm{\beta_1}  \\
		\bm{\beta_2} \\ 
		\bm{\beta_3} \\ 
		\bm{\beta_4}\\ 
		\bm{\beta_5}  \\ 
	\end{array}
	\right).
\end{equation}

It is easy to confirm that $(\bm{\beta_1})$ and $(\bm{\beta_2},\bm{\beta_3},\bm{\beta_4},\bm{\beta_5})$ can generate two additive codes $C_1$ and $C_2$ with parameters $[11,0.5,9]_4$ and $[11,2,8]_4$, respectively, where $C_{[11,2,8]_4}$ has weight distribution $1+15z^8$, and $C_{[11,0.5,9]_4}$ is an invariant subcode of $C_{[11,2.5,8]_4}$. 
It follows that $[11,2.5,8]_4$ is a two-weight code, and the weight distribution is $1+15z^8+16z^9$.
	
	 %Similarly, taking $k = 5$, combining $(32,3.5,24]_4$ and $(11,2.5,8]_4$ yields $(43,3.5,32]_4$ with weight distribution $1+63z^{32}+64z^{33}$, and $(43,3.5,32]_4$ has an invariant subcode $(43,0.5,33]_4$. 
	Thus, for $k\ge 7$, by iterative construction we can obtain all codes with parameters $[ \frac{2^k+1}{3} , \frac{k}{2},2^{k-2}]_4$ with weight distribution $1+(2^{k-1}-1)z^{2^{k-2}}+2^{k-1}z^{2^{k-2}+1}$, which has an invariant subcode with parameters $[ \frac{2^k+1}{3} , 0.5,2^{k-2}+1]_4$. 
	In addition, since $g(k,2^{k-1}+2)=2^k+k > 3(\frac{2^k+1}{3} +1)=2^k+4$, when $k\ge 5$, $[ \frac{2^k+1}{3}, \frac{k}{2},2^{k-2}]_4$ code is GPO. 
\end{proof}
%\begin{lemma}$\star\star\star$
%	When $k\ge 3$, the additive code in Lemma \ref{One-third} is an additive Griesmer code.
%\end{lemma}

\begin{remark}
	Recently, Bierbrauer et al. \cite{Bierbrauer2023AnAP} also obtained this class of additive codes, and Lemma \ref{One-third} can be regarded as a different construction method and rediscover for the same class of codes.
\end{remark}

\begin{lemma}\label{anticode_k_1_k_2+1/3}
	Let $k\ge7$ and $k-3 \ge k_2\ge 4$ be respective odd and even integers. Then, we have the following results.
	
	(1) There exists a four-weight additive GDO code with parameters $ [  \frac{2^k-2^{k_2}+2}{3}  , \frac{k}{2}, 2^{k-2}-2^{k_2-2}]_4$, whose weight distribution is $1+(2^{k-1}-2^{k_1-1})\cdot z^{2^{k-2}-2^{k_2-2}}+(2^{k-1}-2^{k_1-1})\cdot z^{2^{k-2}-2^{k_2-2}+1}+(2^{k_1-1}-1)z^{k-2}+2^{k_1-1}z^{2^{k-2}+1}$, and has an invariant subcode with parameters $[ \frac{2^k-2^{k_2}+2}{3}  , 0.5, 2^{k-2}-2^{k_2-2}+1]_4$;
	 
	 (2) If $k_2=4$, $[ \frac{2^k-2^{k_2}+2}{3}  , \frac{k}{2}, 2^{k-2}-2^{k_2-2}]_4$ code is GDO; otherwise, it is GPO. 
\end{lemma}
\begin{proof}
	By Theorems \ref{anticode_k_1_k_2} and \ref{anticode_k_1_k_2_1/3}, one can construct an additive $[\frac{2^k-2^{k_1}-2^{k_2}+1}{3}, \frac{k}{2}, 2^{k-2}-2^{k_1-2}-2^{k_2-2}]_4$ code has two disjoint subcodes  $[\frac{2^k-2^{k_1}-2^{k_2}+1}{3}, \frac{k_1}{2}, 2^{k-2}-2^{k_1-2}]_4$ and $[\frac{2^k-2^{k_1}-2^{k_2}+1}{3}, \frac{k_2}{2}, 2^{k-2}-2^{k_2-2}]_4$.
%	\begin{equation}
%		\begin{matrix}
%			[\frac{2^k-2^{k_1}-2^{k_2}+1}{3}, \frac{k_1}{2}, 2^{k-2}-2^{k_1-2}]_4\\
%			[\frac{2^k-2^{k_1}-2^{k_2}+1}{3}, \frac{k_2}{2}, 2^{k-2}-2^{k_2-2}]_4
%		\end{matrix}\subset \left[\frac{2^k-2^{k_1}-2^{k_2}+1}{3}, \frac{k}{2}, 2^{k-2}-2^{k_1-2}-2^{k_2-2}\right]_4.
%	\end{equation}

	 In addition, Lemma \ref{One-third} reveals the existence of two-weight $[ \frac{2^{k_1}+1}{3} , \frac{k_1}{2},2^{{k_1}-2}]_4$ code, which has an invariant subcode with parameters $[ \frac{2^{k_1}+1}{3} , 0.5,2^{{k_1}-2}+1]_4$. By virtue of Lemma \ref{construction_X}, we can obtain an additive code with parameters $ [ \frac{2^k-2^{k_2}+2}{3}  , \frac{k}{2}, 2^{k-2}-2^{k_2-2}]_4$, whose generator matrix has the following form
	 \begin{equation}
	 				\setlength{\arraycolsep}{0pt}
	 	\begin{array}{rl}
	 		A_{[ \frac{2^k-2^{k_2}+2}{3}  , \frac{k}{2}]_4}=&\begin{pmatrix}
	 			A_{[\frac{2^k-2^{k_1}-2^{k_2}+1}{3}, \frac{k_1}{2}]_4}&A_{[ \frac{2^{k_1}+1}{3} , \frac{k_1}{2}]_4}\\
	 			A_{[\frac{2^k-2^{k_1}-2^{k_2}+1}{3}, \frac{k_2}{2}]_4}&\bm{0_{k_2\times( \frac{2^{k_1}+1}{3} )}}
	 		\end{pmatrix} \\
	 		 =&\begin{pmatrix}
	 		 	A_{[\frac{2^k-2^{k_1}-2^{k_2}+1}{3}, \frac{k_1}{2}]_4}& \begin{matrix}
	 		 		A_{[ \frac{2^{k_1}+1}{3} , 0.5]_4}\\
	 		 		A_{[ \frac{2^{k_1}+1}{3} , \frac{k_1}{2}-0.5]_4}
	 		 	\end{matrix}\\
	 		 	A_{[\frac{2^k-2^{k_1}-2^{k_2}+1}{3}, \frac{k_2}{2}]_4}&\bm{0_{k_2\times( \frac{2^{k_1}+1}{3} )}}
	 		 \end{pmatrix},
	 	\end{array} 
	 \end{equation}
 	where $A_{[ \frac{2^{k_1}+1}{3} , 0.5]_4}$ is the generator matrix of invariant subcode $[ \frac{2^k-2^{k_2}+2}{3}  , 0.5, 2^{k-2}-2^{k_2-2}+1]_4$.

	Since the $[ \frac{2^{k_1}+1}{3}, \frac{k_1}{2},2^{{k_1}-2}]_4$ code is a two-weight code and has an invariant subcode with parameters $[ \frac{2^{k_1}+1}{3}, 0.5,2^{{k_1}-2}+1]_4$, $A_{[ \frac{2^k-2^{k_2}+2}{3}, \frac{k}{2}]_4}$ can be divided into two parts, the first row is an invariant subcode with parameters $[ \frac{2^k-2^{k_2}+2}{3}, 0.5]_4$, the lower $k-1$ rows are a two-weight code with weight distribution $1+(2^{k-1}-2^{k_1-1})\cdot z^{2^{k-2}-2^{k_2-2}}+(2^{k_1-1}-1)\cdot z^{2^{k-2}}$. Therefore, $A_{[ \frac{2^k-2^{k_2}+2}{3}  , \frac{k}{2}]_4}$ generates a four-weight code with weight distribution $1+(2^{k-1}-2^{k_1-1})\cdot z^{2^{k-2}-2^{k_2-2}}+(2^{k-1}-2^{k_1-1})\cdot z^{2^{k-2}-2^{k_2-2}+1}+(2^{k_1-1}-1)z^{k-2}+2^{k_1-1}z^{2^{k-2}+1}$.

Since $g(k,2^{k-1}-2^{k_2-1}+2)=2^k-2^{k_2}+k_2+1$, we have that $[ \frac{2^k-2^{k_2}+2}{3}  , \frac{k}{2}, 2^{k-2}-2^{k_2-2}]_4$ code is GDO for $k_2=4$, and GPO for $k_2>4$. 
\end{proof}

%\begin{lemma}
%	Let $k=k_1+k_2$, and $k_1,k_2\ge2$, $k,k_1$ be odd, $C_k, C_{k_1}, C_{k_2}$ be additive point set codes, then additive point set partition codes have the following properties.
%	$$\begin{matrix} 
%		\left ( \frac{2^k-2^{k_1}-2^{k_2}+1}{3}, \frac{k_1}{2}, 2^{k-2}-2^{k_1-2} \right ) \\  
%		\left ( \frac{2^k-2^{k_1}-2^{k_2}+1}{3}, \frac{k_2}{2}, 2^{k-2}-2^{k_2-2} \right )
%	\end{matrix}
%	\subset  \left ( \frac{2^k-2^{k_1}-2^{k_2}+1}{3}, \frac{k}{2}, 2^{k-2}-2^{k_1-2}-2^{k_2-2}  \right ), $$
%	
%	$$\left ( \frac{2^k-2^{k_1}}{3}, \frac{k_2}{2}, 2^{k-2} \right )
%	\subset  \left ( \frac{2^k-2^{k_1}}{3}, \frac{k}{2}, 2^{k-2}-2^{k_1-2}  \right ). $$
%\end{lemma}

	\begin{lemma} \label{G_X_enlarger}
		Let $k\ge 5$ be an odd number.
		Suppose that $A_{[ \frac{2^k+1}{3} , \frac{k}{2}]_4}$ and $A_{\frac{k-1}{2}}$ are separately generator matrices of the code in Lemma \ref{One-third} and $\frac{k-1}{2}$-dimensional ASEP code. Then, we have the following results.
		
		(1) The matrix $
			\left( A_{[ \frac{2^k+1}{3} , \frac{k}{2},2^{k-2}]_4}
			\begin{array}{|cc}
				\bm{1_{2^{k-1}-1}} \\ 
				A_{\frac{k-1}{2}}\\ 
			\end{array}
			\right)$
generates a two-weight additive code with parameters $[  \frac{2^k+1}{3}  + 2^{k-1}-1, \frac{k}{2},5\cdot 2^{k-3} ]_4$
%		\begin{equation}
%			(  \frac{2^k+1}{3}  + 2^{k-1}-1, \frac{k}{2},5\cdot 2^{k-3} ]_4,
%		\end{equation}
		 and weight distribution $1+(2^{k}-2)z^{5\cdot 2^{k-3}} + z^{3\cdot 2^{k-2}}$.
		 
		 (2) The additive $[\frac{2^k+1}{3}  + 2^{k-1}-1, \frac{k}{2},5\cdot 2^{k-3} ]_4$ code is both GPO and Griesmer code.
	\end{lemma}
\begin{proof}
	Based on the process of proving Lemmas \ref{aug_Simplex} and \ref{One-third}, the left and right sides of the above matrix generate additive codes with parameters $[\frac{2^{k}+1}{3} , \frac{k}{2},2^{k-2}]_4$ and $[2^{k-1}-1,\frac{k}{2},3\cdot 2^{k-3}-1]_4$, respectively, both of which have two disjoint subcodes. 
	$[\frac{2^{k}+1}{3} , \frac{k}{2},2^{k-2}]_4$ code has disjoint subcodes $[ \frac{2^k+1}{3} , 0.5,2^{k-2}+1]_4$ and $[\frac{2^{k}+1}{3} , \frac{k}{2}-0.5,2^{k-2}]_4$;
	$[2^{k-1}-1,\frac{k}{2},3\cdot 2^{k-3}-1]_4$ code has disjoint subcodes $[2^{k-1}-1,0.5,2^{k-1}-1]_4$ and $[2^{k-1}-1,\frac{k}{2}-0.5,3\cdot 2^{k-3}]_4$, where the $[2^{k-1}-1,\frac{k}{2}-0.5,3\cdot 2^{k-3}]_4$ and $[\frac{2^{k}+1}{3} , \frac{k}{2}-0.5,2^{k-2}]_4$ codes are both constant-weight codes. 
%	there are $\begin{matrix}
%		(2^{k-1}-1,0.5,2^{k-1}-1]_4\\
%		(2^{k-1}-1,\frac{k}{2}-0.5,3\cdot 2^{k-3}]_4
%	\end{matrix}\subset (2^{k-1}-1,\frac{k}{2},3\cdot 2^{k-3}-1]_4$ and $\begin{matrix}
%		[\frac{2^k+1}{3} , 0.5,2^{k-2}+1]_4\\
%		(\frac{2^{k}+1}{3} , \frac{k}{2}-0.5,2^{k-2}]_4
%	\end{matrix}\subset (\frac{2^{k}+1}{3} , \frac{k}{2},2^{k-2}]_4$

	For any nonzero codewords $\bm{c_1}$ of $C_{[2^{k-1}-1,0.5,2^{k-1}-1]_4}$ and $\bm{c_2}$ of $C_{[2^{k-1}-1,\frac{k}{2}-0.5,3\cdot 2^{k-3}]_4}$, we have $\mathbf{w}(\bm{c_1}+\bm{c_2})=3\cdot 2^{k-3}-1$. For for any nonzero codewords $\bm{c_1^{\prime}}$ of $C_{[ \frac{2^k+1}{3} , 0.5,2^{k-2}+1]_4}$ and $\bm{c_2^{\prime}}$ of $C_{[\frac{2^{k}+1}{3} , \frac{k}{2}-0.5,2^{k-2}]_4}$, we get that $\mathbf{w}(\bm{c_1^{\prime}}+\bm{c_2^{\prime}})=3\cdot 2^{k-2}+1$. 
	
	Therefore, the new code $[  \frac{2^k+1}{3}  + 2^{k-1}-1, \frac{k}{2},2^{k-2}+3\cdot 2^{k-3}]_4$ also has two disjoint subcodes $[ \frac{2^k+1}{3}  + 2^{k-1}-1, 0.5,2^{k-2}+2^{k-1}]_4$ and $[ \frac{2^k+1}{3}  + 2^{k-1}-1, \frac{k}{2}-0.5,2^{k-2}+3\cdot 2^{k-3}]_4$, where $[ \frac{2^k+1}{3}  + 2^{k-1}-1, \frac{k}{2}-0.5,2^{k-2}+3\cdot 2^{k-3}]_4$ code has constant-weight.
	
%	satisfy 
%	\begin{equation}
%		\begin{matrix}
%			( \frac{2^k+1}{3}  + 2^{k-1}-1, 0.5,2^{k-2}+2^{k-1}]_4\\
%			( \frac{2^k+1}{3}  + 2^{k-1}-1, \frac{k}{2}-0.5,2^{k-2}+3\cdot 2^{k-3}]_4
%		\end{matrix}\subset (  \frac{2^k+1}{3}  + 2^{k-1}-1, \frac{k}{2},2^{k-2}+3\cdot 2^{k-3}]_4, 
%	\end{equation}
	 Furthermore, for any nonzero codewords 
	$ \bm{c_1^{\prime\prime}}$ of $C_{[ \frac{2^k+1}{3}  + 2^{k-1}-1, 0.5]_4}$ and $\bm{c_2^{\prime\prime}}$ of $C_{[ \frac{2^k+1}{3}  + 2^{k-1}-1, \frac{k}{2}-0.5]_4}$, we can deduce that $\mathbf{w}(\bm{c_1^{\prime\prime}}+\bm{c_2^{\prime\prime}})=2^{k-2}+3\cdot 2^{k-3}$.
	Hence, the $[  \frac{2^k+1}{3}  + 2^{k-1}-1, \frac{k}{2},2^{k-2}+3\cdot 2^{k-3}]_4$ code has weight distribution $1+(2^{k}-2)z^{5\cdot 2^{k-3}} + z^{3\cdot 2^{k-2}}$.
	
	In addition, since $g(k,5\cdot 2^{k-2} )=5\cdot 2^{k-1}-2=3(\frac{2^k+1}{3}  + 2^{k-1}-1)$, and $g(k,5\cdot 2^{k-2}+2 )=5\cdot 2^{k-1}+k-2>3(\frac{2^k+1}{3}  + 2^{k-1})$ for $k\ge5$, $[  \frac{2^k+1}{3}  + 2^{k-1}-1, \frac{k}{2},5\cdot 2^{k-3} ]_4$ code is an additive both GPO and Griesmer code. 
\end{proof}

		\begin{lemma}\label{third_one_k2}
		Let $k\ge 7$ and $k_2\ge 4$ be odd and even numbers, respectively.
	Suppose $A_{[ \frac{2^k-2^{k_2}+2}{3}  , \frac{k}{2}]_4}$ and $A_{\frac{k-1}{2}}$ are generator matrices of the code in Lemma \ref{anticode_k_1_k_2+1/3} and $\frac{k-1}{2}$-dimensional ASEP code, respectively. Then, we have the following results.
		
		(1) The matrix	$
				\left( A_{[ \frac{2^k-2^{k_2}+2}{3}  , \frac{k}{2}]_4}
				\begin{array}{|cc}
					\bm{1_{2^{k-1}-1}} \\ 
					A_{\frac{k-1}{2}}\\ 
				\end{array}
				\right)$
generates a two-weight additive code with parameters
$[  \frac{2^k-2^{k_2}+2}{3}  + 2^{k-1}-1, \frac{k}{2},5\cdot 2^{k-3}-2^{k_2-2} ]_4$
			with weight distribution $1+(2^{k}-2)z^{5\cdot 2^{k-3}-2^{k_2-2}} + z^{3\cdot 2^{k-2}-2^{k_2-2}}$.
			
		(2)	The additive $[  \frac{2^k-2^{k_2}+2}{3}  + 2^{k-1}-1, \frac{k}{2},5\cdot 2^{k-3}-2^{k_2-2} ]_4$ code is both GPO and Griesmer code.
		\end{lemma}
		\begin{proof}
			Similarly, this lemma can be proved using the same method as Lemma \ref{G_X_enlarger}.
		\end{proof}

%	\begin{remark}
%		Compared to the optimal quaternary linear codes listed in \cite{Grassltable}, all the codes in Table \ref{T_Minimal_additive} except for No. 5 have double the number of codewords than the optimal linear codes of same length and distance. Furthermore, compared with the quaternary few-weight codes , our few-weight codes feature entirely distinct parameters.
%		%	In addition, compared with the optimal quaternary codes in \cite{Grassltable}, the dimensions of the codes in Table \ref{T_Minimal_additive} are higher than those in a except for N0. 5.
%	\end{remark}% more than a dozen
So far, we construct ten classes of optimal additive codes and determine their weight distribution. 
%By Lemma \ref{outperform_linear}, we know that additive GPO and Griesmer additive codes with non-integer dimensions outperform linear codes.
For the reader's convenience, we summarize additive GPO and Griesmer additive codes in Table \ref{T_Minimal_additive}.
In particular, we omit some relatively trivial parameters that are not listed, such as $[2^k-1,\frac{k}{2}+1,3\cdot 2^{k-2}-1]_4$ code in Lemma \ref{aug_Simplex} that can be directly derived from No. 2 in Table \ref{T_Minimal_additive}. 
%These codes are superior to optimal linear codes in  \cite{Grassltable,Marutatables} for $k$ is odd. 
%, but the code is an additive Griesmer code, which also outperforms linear codes.
%These codes are superior to linear codes in two ways, on the one hand to optimal linear codes  \cite{Grassltable,Marutatables} and on the other hand to the linear few-weight codes  \cite{Shi2020SeveralFO,Li2023CharacterizationOP}.

\begin{remark}\label{cover_results}
	It should be noted that by the additive codes in Table \ref{T_Minimal_additive} and Lemma \ref{combination_X_construction}, one can cover all the results of additive GPO codes with non-integer dimensions in \cite{bierbrauer2021optimal} and \cite{Guan2023SomeGQ}. Details are as follows.
	%\begin{multicols}{2}%{?}的这个“？”可以是任意数字，代表着你要分几列
\begin{itemize}
	\item   $[8,2.5,6]_4$ (No.2, $k=3$)
\item   $[11,2.5,8]_4$ (No.7, $k=5$)
\item   $[26,2.5,20]_4$ (No.9, $k=5$)
\item   $[31,2.5,24]_4$ (No.1, $k=5$)
\item   $[32,3.5,24]_4$ (No.2, $k=5$)
\item   $[35,3.5,26]_4$ (No.6, $k=7$, $k_1=3$)
\item   $[40,3.5,30]_4$ (No.5, $k=7$, $k_1=3$)
\item   $[43,3.5,32]_4$ (No.7, $k=7$)
\item   $[127,3.5,96]_4$ (No.1, $k=7$)
\item   $[128,4.5,96]_4$ (No.2, $k=7$)	
\item   $[155,4.5,116]_4$ (No.6, $k=9$, $k_1=5$)
\item   $[160,4.5,120]_4$ (No.5, $k=9$, $k_1=5$)
\item   $[168,4.5,126]_4$ (No.5, $k=9$, $k_1=3$)
\item   $[171,4.5,128]_4$ (No.7, $k=9$)
\item $[163,4.5,122]_4$ (using $[35,3.5,26]_4$ code by Lemma \ref{combination_X_construction})
\end{itemize}
%\end{multicols}

% $(8,2.5,6]_4$, $(11,2.5,8]_4$, $(26,2.5,20]_4$, $(31,2.5,24]_4$ non-trivial optimal quaternary additive codes \cite{bierbrauer2021optimal}. 

	Specifically, additive $[8,2.5,6]_4$, $[11,2.5,8]_4$, $[26,2.5,20]_4$ and $[31,2.5,24]_4$ codes are the  quaternary additive GPO codes constructed in \cite{bierbrauer2021optimal}, while the others are results in \cite{Guan2023SomeGQ}. 
\end{remark}

	\begin{table*}[ht]
		\caption{Quaternary additive Griesmer puncture-optimal codes with $k=k_1+k_2$}\label{T_Minimal_additive}
		\centering
		\resizebox{\textwidth}{!}{
			\begin{threeparttable}
				\begin{tabular}{ccccc}
					\toprule
					No.     &                                                 Additive Codes                                                  &                                                                                                Weight Distribution                                                                                                &                   Constructions                   &                      Ranges                       \\ \midrule
					\textbf{1}  &                                  $\bm{[2^{k}-1,\frac{k}{2} ,3\cdot2^{k-2}]_4}^\diamond$                                  &                                                                                        $\bm{1+2^k\cdot z^{3\cdot2^{k-2}}}$                                                                                        &       \textbf{Lemma} \ref{APS_construction}       &                   \bm{$k\ge 3$}                   \\ \hline
					%				2  &                                 $[2^k-1,\frac{k}{2}+1,3\cdot 2^{k-2}-1]_4^\diamond$                                 &                                              \makecell[c]{$1+(2^{k+2}-2^k-3)\cdot z^{3\cdot 2^{k-2}-1}+$\\ ${(2^k-1)\cdot z^{3\cdot 2^{k-2}}+3z^{2^k-1}}$}                                              &          Lemma \ref{aug_Simplex}         &                   $k\ge 3$                   \\ \hline
					\textbf{2}  &                                   $\bm{[2^k,\frac{k}{2}+1,3\cdot 2^{k-2}]_4}^\diamond$                                   &                                                                               $\bm{1+(2^{k+2}-4)\cdot z^{3\cdot 2^{k-2}}+3z^{2^k}}$                                                                               &         \textbf{Lemma} \ref{aug_Simplex}          &                   \bm{$k\ge 3$}                   \\
					\hline
					\textbf{3}  &                  $\bm{[2^k-1-\frac{2^{k_2}-1}{3}, \frac{k}{2}, 3\cdot 2^{k-2}- 2^{k_2-2}]_4}^\diamond$                   &                                                \makecell[c]{$\bm{1+(2^k-2^{k_1})\cdot z^{3\cdot 2^{k-2}-2^{k_2-2}}}$\\$\bm{+(2^{k_1}-1)\cdot z^{3\cdot2^{k-2}}}$}                                                 & \textbf{Corollary} \ref{anticode_k_1_divided_k_2} &    \makecell[c]{\bm{$k\ge 7$}\\\bm{$k_2\ge4$}}    \\ \hline
					\textbf{4}      &              $\bm{[2^k-1-m\cdot \frac{2^{k_2}-1}{3}, \frac{k}{2}, 3\cdot 2^{k-2}-m\cdot 2^{k_2-2}]_4}$               &                                                  \makecell[c]{$\bm{1+(2^k-2^{k_1})\cdot z^{3\cdot 2^{k-2}-m\cdot2^{k_2-2}}}$\\$\bm{+(2^{k_1}-1)\cdot z^{3\cdot2^{k-2}}}$}                                                   &     \textbf{Corollary} \ref{anticode_k_1_divided_k_2}      & \makecell[c]{$\bm{k\ge 7}$\\$\bm{k_2\ge4}$\\$\bm{2\le m \le 3}$} \\ \hline
					%				     6      &   $[2^k-2^{k_1}-m\cdot \frac{2^{k_2}-1}{3}, \frac{k}{2}, 3\cdot 2^{k-2}-3\cdot 2^{k_1-2}-m\cdot 2^{k_2-2}]_4$   & \makecell[c]{$1+(2^k-2^{k_1}-2^{k_2}+1)\cdot z^{3\cdot 2^{k-2}-3\cdot 2^{k_1-2}-m\cdot 2^{k_2-2}}$\\$+(2^{k_1}-1)\cdot z^{3\cdot 2^{k-2}-3\cdot 2^{k_1-2}}+(2^{k_2}-1)\cdot z^{3\cdot 2^{k-2}-m\cdot 2^{k_2-2}}$} &     Corollary \ref{anticode_k_1_divided_k_2}      & \makecell[c]{$k\ge 7$\\$k_1\ge3$\\$1\le m \le 2$} \\ \hline
					\textbf{5}  &                        $\bm{[ \frac{2^k-2^{k_1}}{3}, \frac{k}{2}, 2^{k-2}-2^{k_1-2}]_4}^\diamond$                        &                                                                 $\bm{1+(2^k-2^{k_2})\cdot z^{ 2^{k-2}- 2^{k_1-2}}+(2^{k_2}-1)\cdot z^{ 2^{k-2}}}$                                                                 &    \textbf{Theorem} \ref{anticode_k_1_k_2_1/3}    &        \makecell[c]{$\bm{k_1\ge 3}$\\$\bm{k_2\ge4}$}        \\ \hline
					\textbf{6}  &             $\bm{ [ \frac{2^k-2^{k_1}-2^{k_2}+1}{3}, \frac{k}{2}, 2^{k-2}-2^{k_1-2}-2^{k_2-2}]_4}^\diamond$              &                  \makecell[c]{\bm{$1+(2^k-2^{k_1}-2^{k_2})\cdot z^{ 2^{k-2}- 2^{k_1-2}- 2^{k_2-2}}+$}\\\bm{$(2^{k_1}-1)\cdot z^{ 2^{k-2}- 2^{k_1-2}}+(2^{k_2}-1)\cdot z^{ 2^{k-2}- 2^{k_2-2}}$}}                  &    \textbf{Theorem} \ref{anticode_k_1_k_2_1/3}    &   \makecell[c]{\bm{$k_1\ge 3$}\\\bm{$k_2\ge4$}}   \\ \hline
					\textbf{7}      &                            $\bm{[ \frac{2^k+1}{3} , \frac{k}{2},2^{k-2}]_4}$                             &                                                                                  $\bm{1+(2^{k-1}-1)z^{2^{k-2}}+2^{k-1}z^{2^{k-2}+1}}$                                                                                  &               \textbf{Lemma} \ref{One-third}               &                     $\bm{k\ge 5}$                      \\ \hline
					%				    10      &            $[ \frac{2^k-2^{k_2}+2}{3}  , \frac{k}{2}, 2^{k-2}-2^{k_2-2}]_4$             &                       \makecell[c]{$1+(2^{k-1}-2^{k_1-1})\cdot z^{2^{k-2}-2^{k_2-2}}+$\\$(2^{k-1}-2^{k_1-1})\cdot z^{2^{k-2}-2^{k_2-2}+1}$\\$+(2^{k_1-1}-1)z^{k-2}+2^{k_1-1}z^{2^{k-2}+1}$}                       &         Lemma \ref{anticode_k_1_k_2+1/3}          &         \makecell[c]{$k\ge 7$\\$k_2=4$}         \\ \hline
					\textbf{8}      &            $\bm{[ \frac{2^k-2^{k_2}+2}{3}  , \frac{k}{2}, 2^{k-2}-2^{k_2-2}]_4}$             &                       \makecell[c]{$\bm{1+(2^{k-1}-2^{k_1-1})\cdot z^{2^{k-2}-2^{k_2-2}}+}$\\$\bm{(2^{k-1}-2^{k_1-1})\cdot z^{2^{k-2}-2^{k_2-2}+1}}$\\$\bm{+(2^{k_1-1}-1)z^{k-2}+2^{k_1-1}z^{2^{k-2}+1}}$}                       &         \textbf{Lemma} \ref{anticode_k_1_k_2+1/3}          &         \makecell[c]{$\bm{k\ge k_2+3}$\\$\bm{k_2\ge6}$}         \\ \hline
					\textbf{9} &               $\bm{[  \frac{2^k+1}{3}  + 2^{k-1}-1, \frac{k}{2},5\cdot 2^{k-3}  ]_4}^\diamond$               &                                                                             $\bm{1+(2^{k}-2)z^{5\cdot 2^{k-3}} + z^{3\cdot 2^{k-2}}}$                                                                             &         \textbf{Lemma} \ref{G_X_enlarger}         &                   \bm{$k\ge 5$}                   \\ \hline
					\textbf{10} & $\bm{[  \frac{2^k-2^{k_2}+2}{3}  + 2^{k-1}-1, \frac{k}{2},5\cdot 2^{k-3}-2^{k_2-2}]_4}^\diamond$ &                                                                   $\bm{1+(2^{k}-2)z^{5\cdot 2^{k-3}-2^{k_2-2}} + z^{3\cdot 2^{k-2}-2^{k_2-2}}}$                                                                   &         \textbf{Lemma} \ref{third_one_k2}         &                  \bm{$k\ge 7$ }                   \\ \bottomrule
				\end{tabular}
				\begin{tablenotes}    %这行要添加， 从这开始
					\footnotesize               %这行要添加are optimal or distance-optimal, and those 
					\item[] \textbf{Note:} According to Lemma \ref{Griesmer_Bound}, additive codes in this table are all additive GPO codes. The superscript ``$\diamond$" indicates that the corresponding is also an additive Griesmer code.
					%				Recently, the authors of \cite{Bierbrauer2023AnAP} also independently obtained Nos. 1 and 8 via geometrical approaches. 
					%				In Nos. 1-11, if $k$ and $k_1$ or $k_2$ appear simultaneously, then $k$ is assumed to be $k_1 + k_2$ and greater or equal to $7$.
					%				
					%				$k=k_1$ and $k_1\ge 3$ are odd, $k_2=k-k_1$, $1\le m\le 3$. According to Lemma \ref{Griesmer_Bound}, all additive codes in this table are optimal or distance-optimal, and those marked in bold indicate that they are Griesmer codes. 			
				\end{tablenotes}            %这行要添加
			\end{threeparttable}  
		}  
	\end{table*}
	
\section{Constructions of 3.5-dimensional optimal quaternary additive codes}\label{Sec_3.5_4.5}
		In this section, we employ the additive codes obtained in the previous sections and associated construction methods to construct quaternary 3.5-dimensional optimal additive codes.

	\begin{definition}\label{X-optimal}
		Let $C_1$ be a linear optimal $[n,k,d]_4^l$ code. If there exists a linear optimal code $C_2$ with parameters $[n,k-1,d+1]^l_4$ satisfying $C_2\subset C_1$, then we call $C_1$ a \textit{linear}  $X$\textit{-optimal code}.
	\end{definition}

	\begin{lemma}\label{Generalized_X_Construction_Corollary}
	If there exists a quaternary linear $X$-optimal code $C_x$ with parameters $[n_x,k,d_x]_4^l$, then there also exist additive codes with parameters $[\frac{2^{2k-1}+1}{3} +n_x, k-0.5,2^{2k-3}+d_x+1]_4$ and $ [\frac{2^{2k-1}-2^{k_2}+2}{3}+n_x, k-0.5, 2^{2k-3}-2^{k_2-2}+d_x+1]_4$, satisfying $k_2< 2k-1$, and $k_2$ is an even number.
	\end{lemma}
 {	\begin{proof}
	Since all linear codes are additive, $C_x$ also can be written as an additive code $[n,k,d]_4$.
	According to Definition \ref{X-optimal}, $C_x$ has an additive subcode $C_x^\prime$ with parameters  $[n,k-0.5,d]_4$, satisfy $[n,k-1,d+1]_4 \subset [n,k-0.5,d]_4$.
	From Lemmas \ref{One-third} and \ref{anticode_k_1_k_2+1/3}, we have two classes of additive codes with parameters $[ \frac{2^k+1}{3} , \frac{k}{2},2^{k-2}]_4$ and $ [  \frac{2^k-2^{k_2}+2}{3}  , \frac{k}{2}, 2^{k-2}-2^{k_2-2}]_4$, and invariant subcodes of them have parameters $[\frac{2^k+1}{3} , 0.5,2^{k-2}+1]_4$ and $[\frac{2^k-2^{k_2}+2}{3}  , 0.5, 2^{k-2}-2^{k_2-2}+1]_4$, respectively.
		By Corollary \ref{cor_Generalized_X_Construction}, the combination of $C_x$ and these two classes of additive codes can produce additive codes with parameters $[\frac{2^{2k-1}+1}{3} +n_x, k-0.5,2^{2k-3}+d_x+1]_4$ and $ [\frac{2^{2k-1}-2^{k_2}+2}{3}+n_x, k-0.5, 2^{2k-3}-2^{k_2-2}+d_x+1]_4$.
	\end{proof}}

	\begin{example}\label{Combine_optimal_codes}
		In \cite{Grassltable}, there is an $X$-optimal code with parameters $[16,4,11]_4^l$, which satisfies $[16,3,12]_4^l\subset [16,4,11]_4^l$. 
		Write the generator matrix of $X$-optimal $[16,4,11]_4^l$ code in the additive form (it is an $8\times16$ matrix) and choose the lower $7$-rows as $A_{[16,3.5]_4}$, such that the lower 6-rows of $A_{[16,3.5]_4}$ generates an additive $[16,3,12]_4$ code.
%		Construct a matrix $A_{(16,3.5]_4}$, whose first row in $(16,4,11]_4\setminus(16,3,12]_4$, and the lower 6-rows can generate an additive code with parameters $(16,3,12]_4$.
		\begin{equation}\label{(3,1.5,2)}
				\setlength{\arraycolsep}{0pt}
{\scriptsize 		A_{[16,3.5]_4}= \left( {\begin{array}{*{16}{c}}
				0 & 0 & 0 & 1   & 1   & w   & 0   & 1   & w   & w^2 & w^2 & w & 1   & 0   & w   & 1   \\
				1 & 0 & 0 & 1   & 0   & w^2 & w   & 1   & w^2 & 1   & 0   & 1 & w^2 & 1   & w   & w^2 \\
				w & 0 & 0 & w   & 0   & 1   & w^2 & w   & 1   & w   & 0   & w & 1   & w   & w^2 & 1   \\
				0 & 1 & 0 & w^2 & w   & 1   & w^2 & 1   & 0   & 1   & w^2 & 1 & w   & w^2 & 0   & 1   \\
				0 & w & 0 & 1   & w^2 & w   & 1   & w   & 0   & w   & 1   & w & w^2 & 1   & 0   & w   \\
				0 & 0 & 1 & 1   & w   & 0   & 1   & w   & w^2 & w^2 & w   & 1 & 0   & w   & 1   & 1   \\
				0 & 0 & w & w   & w^2 & 0   & w   & w^2 & 1   & 1   & w^2 & w & 0   & w^2 & w   & w
		\end{array}} \right)}.
		\end{equation}
%			\begin{equation}\label{(3,1.5,2)}
%		\setlength{\arraycolsep}{0.5pt}
%		G_{(16,3.5]_4}= \left( {\begin{array}{*{16}{c}}
%				0 & 0 & 0 & 1   & 1   & w   & 0   & 1   & w   & w^2 & w^2 & w & 1   & 0   & w   & 1   \\
%				1 & 0 & 0 & 1   & 0   & w^2 & w   & 1   & w^2 & 1   & 0   & 1 & w^2 & 1   & w   & w^2 \\
%				0 & 1 & 0 & w^2 & w   & 1   & w^2 & 1   & 0   & 1   & w^2 & 1 & w   & w^2 & 0   & 1   \\
%				0 & 0 & 1 & 1   & w   & 0   & 1   & w   & w^2 & w^2 & w   & 1 & 0   & w   & 1   & 1   \\
%		\end{array}} \right),
%	\end{equation}
		 	\begin{figure*}[htbp] %hb代表放在文章底部，%ht为放在文章顶部 
		 	\centering
	\begin{equation}\label{(38,3.5,28]_4}
	{\tiny 			\setlength{\arraycolsep}{0pt}
		A_{[38,3.5,28]_4}= \left( {\begin{array}{*{43}{c}}
			1 & 1 & w^2 & 0   & w & w^2 & w   & 1 & 1 & w^2 & 0   & w & w^2 & w   & 1   & 1   & w^2 & 0   & w   & w^2 & w   & 1   & 1   & w^2 & 0   & w   & w^2 & w   & 1   & 1   & w^2 & 0   & w   & w^2 & w   & w & w & w \\
			w & 1 & 1   & w^2 & 0 & w   & w^2 & w & 1 & 1   & w^2 & 0 & w   & w^2 & w   & 1   & 1   & w^2 & 0   & w   & w^2 & w   & 1   & 1   & w^2 & 0   & w   & w^2 & w   & 1   & 1   & w^2 & 0   & w   & w^2 & 1 & 0 & 1 \\
			0 & w & w^2 & w   & 1 & 1   & w^2 & 0 & w & w^2 & w   & 1 & 1   & w^2 & 0   & w   & w^2 & w   & 1   & 1   & w^2 & 0   & w   & w^2 & w   & 1   & 1   & w^2 & 0   & w   & w^2 & w   & 1   & 1   & w^2 & 0 & 1 & 1 \\
			1 & 1 & 1   & 1   & 1 & 1   & 1   & 0 & 0 & 0   & 0   & 0 & 0   & 0   & 1   & 1   & 1   & 1   & 1   & 1   & 1   & w   & w   & w   & w   & w   & w   & w   & w   & w   & w   & w   & w   & w   & w   & 0 & 0 & 0 \\
			w & w & w   & w   & w & w   & w   & 0 & 0 & 0   & 0   & 0 & 0   & 0   & w   & w   & w   & w   & w   & w   & w   & w^2 & w^2 & w^2 & w^2 & w^2 & w^2 & w^2 & w^2 & w^2 & w^2 & w^2 & w^2 & w^2 & w^2 & 0 & 0 & 0 \\
			0 & 0 & 0   & 0   & 0 & 0   & 0   & 1 & 1 & 1   & 1   & 1 & 1   & 1   & w   & w   & w   & w   & w   & w   & w   & w   & w   & w   & w   & w   & w   & w   & 1   & 1   & 1   & 1   & 1   & 1   & 1   & 0 & 0 & 0 \\
			0 & 0 & 0   & 0   & 0 & 0   & 0   & w & w & w   & w   & w & w   & w   & w^2 & w^2 & w^2 & w^2 & w^2 & w^2 & w^2 & w^2 & w^2 & w^2 & w^2 & w^2 & w^2 & w^2 & w   & w   & w   & w   & w   & w   & w   & 0 & 0 & 0
		\end{array}} \right).}
\end{equation}		
	\end{figure*}
		 	\begin{figure*}[htbp] %hb代表放在文章底部，%ht为放在文章顶部 
	\centering
				\begin{equation}\label{(43,3.5,32]_4}
{\tiny 			\setlength{\arraycolsep}{0pt}
			A_{[43,3.5,32]_4}= \left( {\begin{array}{*{43}{c}}
				0 & 1 & 0 & 0 & 1   & w   & w^2 & w^2 & 0   & 0   & w^2 & 1 & 1   & 1   & w^2 & 0   & w   & 0   & w^2 & w^2 & w & w^2 & 1   & w^2 & w   & 1   & w & 1   & 0   & w & w   & w   & 0 & 1 & 0 & w^2 & w^2 & 1   & w   & w & w & w & w \\
				0 & w & 0 & 1 & w   & 1   & w   & w^2 & 1   & 1   & w^2 & w & w^2 & w^2 & w   & 1   & 0   & 0   & w^2 & w^2 & 0 & w   & w   & w   & 0   & w^2 & 1 & w^2 & 0   & 0 & 1   & 1   & 0 & w & 1 & w^2 & w   & w^2 & 0   & 1 & 1 & 1 & 0 \\
				0 & 0 & 1 & 0 & w   & 0   & 1   & w   & 1   & w   & w^2 & 1 & 1   & w^2 & w   & w^2 & w^2 & w   & 1   & 0   & 0 & w^2 & w^2 & 0   & w   & w   & w & 0   & w^2 & 1 & w^2 & 1   & 0 & 0 & w & 1   & w^2 & w   & w^2 & 1 & 1 & 0 & 1 \\
				0 & 0 & w & 1 & w^2 & 1   & w   & w   & w^2 & w   & 0   & w & w^2 & 0   & w^2 & 0   & 1   & w^2 & w^2 & 1   & 0 & 1   & 1   & 0   & w^2 & w   & w & 1   & 1   & w & 0   & w^2 & 1 & 1 & 1 & 1   & 1   & 1   & 1   & 1 & 0 & 0 & 0 \\
				0 & 0 & 0 & w & 1   & w^2 & 1   & w   & w   & w^2 & w   & 0 & w   & w^2 & 0   & w^2 & 0   & 1   & w^2 & w^2 & 1 & 0   & 1   & 1   & 0   & w^2 & w & w   & 1   & 1 & w   & w^2 & w & w & w & w   & w   & w   & w   & w & 0 & 0 & 0 \\
				1 & 1 & 1 & 1 & 1   & 1   & 1   & 1   & 1   & 1   & 1   & 1 & 1   & 1   & 1   & 1   & 1   & 1   & 1   & 1   & 1 & 1   & 1   & 1   & 1   & 1   & 1 & 1   & 1   & 1 & 1   & 1   & 0 & 0 & 0 & 0   & 0   & 0   & 0   & 0 & 0 & 0 & 0 \\
				w & w & w & w & w   & w   & w   & w   & w   & w   & w   & w & w   & w   & w   & w   & w   & w   & w   & w   & w & w   & w   & w   & w   & w   & w & w   & w   & w & w   & w   & 0 & 0 & 0 & 0   & 0   & 0   & 0   & 0 & 0 & 0 & 0
			\end{array}} \right).}
		\end{equation}
		\end{figure*}
	
		 By Lemma \ref{Generalized_X_Construction_Corollary}, 
		 combining $A_{[16,3.5]_4}$ with the generator matrices $A_{[38,3.5,28]_4}$ of $[38,3.5,28]_4$ code in Lemma \ref{anticode_k_1_k_2+1/3} and $A_{[43,3.5,32]_4}$ of $[43,3.5,32]_4$ code  in Lemma \ref{One-third}, we can get additive GPO codes with parameters $[54,3.5,40]_4$ and $[59,3.5,44]_4$, respectively. 

		  In particular, using Magma \cite{bosma1997magma}, one can determine that both of them are two-weight codes whose weight distributions are $1+101z^{40}+26z^{44}$ and $1+108z^{44}+19z^{48}$, respectively.
	\end{example}
	
	Next, let us solve the problem of constructing all optimal quaternary $3.5$-dimensional additive codes. As mentioned in Sect. I, this is to construct all GPO or PO codes with dimension $3.5$. 
	Firstly, Lemma \ref{Griesmer_Bound} plays a vital role in determining the upper bound on GPO codes. 
	However, the additive Griesmer bound is typically not tight for codes with short lengths. Here, we use inverse methods to determine the upper bound on PO codes, as shown in the following lemma.
	
	\begin{lemma}\label{L_anti_bound}
		If there is no linear $[3n,2k,2d]^l_2$ code, then there is no additive $[n,k,d]_4$ code.
	\end{lemma}
	\begin{proof}
		If additive $[n,k,d]_4$ code exists, then by concatenating with linear $[3,2,2]_2^l$ code,  we can derive a linear $[3n,2k,2d]^l_2$ code.
		Therefore, if $[3n,2k,2d]^l_2$ code does not exist, $[n,k,d]_4$ does not exist either.
	\end{proof}

%	\begin{lemma}\label{L_two-step-Grisemer bound}
%		If there is no additive $[n-d,k-1,\lfloor \frac{1}{2} \lceil\frac{d}{2}\rceil \rfloor]_4$ code, then there is no additive $[n,k,d]_4$ code.
%	\end{lemma}
%	\begin{proof}
%		The concatenation of codes $[n,k,d]_4$ and $[3,2,2]_2^l$ yields $[3n,2k,2d]_2^l$ code, and  two-step residual code of $[3n,2k,2d]_2^l$ provides $[3n-3d,2k-2,\lceil\frac{d}{2}\rceil]_2$. By Lemma \ref{L_anti_bound}, the existence of $[3n-3d,2k-2,\lceil\frac{d}{2}\rceil]_2$ is the necessary condition for  the existence of $[n-d,k-1,\lfloor \frac{1}{2} \lceil\frac{d}{2}\rceil \rfloor]_4$. Therefore, we finish the proof.
%	\end{proof}
%	
	
	\begin{corollary}\label{C_non-exist}
		There are no quaternary additive codes with the following parameters. \par
		(1) For $n=18,19$, additive $[n,7/2,n-5]_4$ code;\par
		(2) For $n=26,27$, additive $[n,7/2,n-7]_4$ code.\par
	\end{corollary}
 {	\begin{proof}
		These codes satisfy Lemma \ref{Griesmer_Bound}, but according to \cite{Grassltable}, there is no $[54,7,26]^l_2$, $[57,7,28]^l_2$, $[78,7,38]^l_2$ and $[81,7,40]^l_2$ codes. Therefore, the above additive codes violate Lemma \ref{L_anti_bound}.
		%According to \cite{Grassltable}, there is no $[54,7,26]_2$, $[57,7,28]_2$, $[78,7,38]_2$ and $[81,7,40]_2$. 
	\end{proof}}
	
	\begin{lemma}\label{Additive code combinations} Let $s$ be a positive integer. 
		If the additive $[n,\frac{k}{2},d]_4$ code is GPO (resp. GDO), then the additive $[n+s(2^k-1),\frac{k}{2},d+3s\cdot 2^{k-2}]_4$ code is also  GPO (resp. GDO).
	\end{lemma}
	\begin{proof}
		Since $[n,\frac{k}{2},d]_4$ code is an GPO, we have that additive $[n+1,\frac{k}{2},d+1]_4$ code violates Lemma \ref{Griesmer_Bound}, i.e., $3n+3>\sum\limits_{i=0}^{k-1}\left\lceil\frac{d+1}{2^{i-1}}\right\rceil$. 
		%Horizontaljoin additive $[n,\frac{k}{2},d]_4$ code with ASEP $(2^k-1,\frac{k}{2},3\cdot 2^{k-2}]_4$ code, one can obtain additive $[n+2^k-1,\frac{k}{2},d+3\cdot 2^{k-2}]_4$ code.
		As $\sum\limits_{i=0}^{k-1}\left\lceil\frac{d+3s\cdot 2^{k-2}}{2^{i-1}}\right\rceil=2d+3s\cdot 2^{k-1}+d+3s\cdot 2^{k-2}+\cdots+\left\lceil\frac{d}{2^{k-2}}\right\rceil+3s=\sum\limits_{i=0}^{k-1}\left\lceil\frac{d}{2^{i-1}}\right\rceil+3s(2^{k}-1)$, additive $[n+s2^k,\frac{k}{2},d+3s\cdot 2^{k-2}+1]_4$ code also violates Lemma \ref{Griesmer_Bound}, i.e., additive $[n+s(2^k-1),\frac{k}{2},d+3s\cdot 2^{k-2}]_4$ code is also an GPO code.
        By using a similar approach, it can be shown that this result holds for GDO code as well.
%		
%		 $3\cdot(2^{k}-1)=\sum\limits_{i=0}^{k-1}\left\lceil\frac{3\cdot2^{k-2}}{2^{i-1}}\right\rceil=3\cdot\sum\limits_{i=0}^{k-1}2^{k-i-1}$, so there is no $[n+2^k-1+1,\frac{k}{2},d+3\cdot 2^{k-2}+1]_4$, i.e.
%		$[n+2^k-1,\frac{k}{2},d+3\cdot 2^{k-2}]_4$ is an additive Griesmer optimal code.
	\end{proof}
	
 {	 With the help of Lemma \ref{Additive code combinations}, we can transform the problem of determining all quaternary $\frac{k}{2}$-dimensional long GPO codes into the problem of constructing all GPO codes of lengths between $i(2^k-1)$ and $(i+1)(2^k-1)$, where $i$ is a non-negative integer.
	Here, we call $2^k-1$ as the \textit{period} of quaternary $\frac{k}{2}$-dimensional additive codes.} 
	
 {	
	By corollary \ref{C_non-exist}, there are no $3.5$-dimensional GPO codes of lengths $19$ and $27$. Thus, to solve the construction problem of all 3.5-dimensional optimal additive codes, we need to construct at least all GPO or PO codes in the first two periods, i.e., determine all 3.5-dimensional GPO and PO codes with lengths up to $254$. 
	In Table \ref{3.5_additive}, we list them of lengths between $4$ and $254$, except for three undetermined ones.}
	
 	In Table \ref{3.5_additive}, we use $C_t$ denote additive $[n,3.5,n-t]_4$ codes with variable $n$, and $(C_{t_1}\mid C_{t_2})$ denote the juxtaposition code of $C_{t_1}$ and $C_{t_2}$, in which $C_{t_1}$ and $C_{t_2}$ are defaulted to the maximum lengths.
	$C_t$ with the largest $n$ is GPO or PO code, others are GDO or DO codes.
	In Table \ref{3.5_additive}, if $t=5$, $6$ or $\ge 8$, then the code $C_t$ with largest $n$ are all GPO and the others are PO.
	Under Lemma \ref{Additive code combinations}, for $n\ge255$, we can construct all GPO codes, except for $t=6,7,12$, by combining the codes in Table \ref{3.5_additive} with $[127,3.5,96]_4$ code in Lemma \ref{APS_construction}, and all GDO codes can be obtained by puncturing them. 
	
% {	Specifically, we determine optimal additive $[n,3.5,n-t]_4$ codes for all $t$ with variable $n$, except for $t=6,7,12$.
%	This is, if one can construct three GPO or PO codes with parameters: $[22,3.5,16]_4$, $[25,3.5,18]_4$ and $[46,3.5,34]_4$, then optimal parameters of quaternary $3.5$-dimensional additive codes will be determined entirely.}

\begin{table*}[htbp]
	\caption{Optimal quaternary additive 3.5-dimensional codes of lengths from $4$ to $254$}\label{3.5_additive}
	\centering
	\renewcommand\arraystretch{0.75}
	\begin{threeparttable}
		\begin{tabular}{ccccc}
			\toprule
			No.     &    $C_{t}$    &      parameters       &          Range          &                   Constructions                   \\ \midrule
			1      &    $C_{3}$    &    $[n,3.5,n-3]_4$    &     $4\le n \le 7$      &             \cite{blokhuis2004small}              \\
			2      &    $C_{4}$    &    $[n,3.5,n-4]_4$    &     $8\le n \le 12$     &             \cite{blokhuis2004small}              \\
			3      &    $C_{5}$    &    $[n,3.5,n-5]_4$    &    $13\le n \le 17$     &                \cite{Grassltable}                 \\
			4      &    $C_{6}$    &    $[n,3.5,n-6]_4$    &    $18\le n \le 22$     &                         -                         \\
			5      &    $C_{7}$    &    $[n,3.5,n-7]_4$    &    $23\le n \le 25$     &                         -                         \\
			6      &    $C_{8}$    &    $[n,3.5,n-8]_4$    &    $26\le n \le 32$     &              Lemma \ref{aug_Simplex}              \\
			7      &    $C_{9}$    &    $[n,3.5,n-9]_4$    &    $33\le n \le 35$     &        Theorem \ref{anticode_k_1_k_2_1/3}         \\
			8      &   $C_{10}$    &   $[n,3.5,n-10]_4$    &    $36\le n \le 40$     &        Theorem \ref{anticode_k_1_k_2_1/3}         \\
			9      &   $C_{11}$    &   $[n,3.5,n-11]_4$    &    $41\le n \le 43$     &               Lemma \ref{One-third}               \\
			10      &   $C_{12}$    &   $[n,3.5,n-12]_4$    &    $44\le n \le 46$     &                         -                         \\
			11      &   $C_{13}$    &   $[n,3.5,n-13]_4$    &    $47\le n \le 51$     &                         \cite{kurz2024computer}                         \\
			\textbf{12} & $\bm{C_{14}}$ & $\bm{[n,3.5,n-14]_4}$ &  $\bm{52\le n \le 54}$  &   \textbf{Example} \ref{Combine_optimal_codes}    \\
			\textbf{13} & $\bm{C_{15}}$ & $\bm{[n,3.5,n-15]_4}$ &  $\bm{55\le n \le 59}$  &   \textbf{Example} \ref{Combine_optimal_codes}    \\
			14      &   $C_{16}$    &   $[n,3.5,n-16]_4$    &    $60\le n \le 64$     &              Lemma \ref{aug_Simplex}              \\
			15      &   $C_{17}$    &   $[n,3.5,n-17]_4$    &    $65\le n \le 67$     &                $(C_{8}\mid C_{9})$                \\
			16      &   $C_{18}$    &   $[n,3.5,n-18]_4$    &    $68\le n \le 72$     &               $(C_{8}\mid C_{10})$                \\
			17      &   $C_{19}$    &   $[n,3.5,n-19]_4$    &    $73\le n \le 75$     &               $(C_{8}\mid C_{11})$                \\
			18      &   $C_{20}$    &   $[n,3.5,n-20]_4$    &    $76\le n \le 80$     &               $(C_{10}\mid C_{10})$               \\
			19      &   $C_{21}$    &   $[n,3.5,n-21]_4$    &    $81\le n \le 85$     &              Lemma \ref{combination_X_construction}              \\
			20      &   $C_{22}$    &   $[n,3.5,n-22]_4$    &    $86\le n \le 86$     &                  Extend $C_{21}$                  \\
			\textbf{21} & $\bm{C_{23}}$ & $\bm{[n,3.5,n-23]_4}$ &  $\bm{87\le n \le 91}$  &             $\bm{(C_{8}\mid C_{15})}$             \\
			22      &   $C_{24}$    &   $[n,3.5,n-24]_4$    &    $92\le n \le 96$     &              Lemma \ref{aug_Simplex}              \\
			\textbf{23} & $\bm{C_{25}}$ & $\bm{[n,3.5,n-25]_4}$ & $\bm{97\le n \le 101}$  &         \textbf{Lemma} \ref{third_one_k2}         \\
			\textbf{24} & $\bm{C_{26}}$ & $\bm{[n,3.5,n-26]_4}$ & $\bm{102\le n \le 106}$ &         \textbf{Lemma} \ref{G_X_enlarger}         \\
			25      &   $C_{27}$    &   $[n,3.5,n-27]_4$    &   $107\le n \le 107$    &                  Extend $C_{26}$                  \\
			26      &   $C_{28}$    &   $[n,3.5,n-28]_4$    &   $108\le n \le 112$    &               Corollary \ref{anticode_k_1_divided_k_2}                \\
			\textbf{27} & $\bm{C_{29}}$ & $\bm{[n,3.5,n-29]_4}$ & $\bm{113\le n \le 117}$ & \textbf{Corollary} \ref{anticode_k_1_divided_k_2} \\
			\textbf{28} & $\bm{C_{30}}$ & $\bm{[n,3.5,n-30]_4}$ & $\bm{118\le n \le 122}$ & \textbf{Corollary} \ref{anticode_k_1_divided_k_2} \\
			29      &   $C_{31}$    &   $[n,3.5,n-31]_4$    &   $123\le n \le 127$    &           Lemma \ref{APS_construction}            \\
			30      &   $C_{32}$    &   $[n,3.5,n-32]_4$    &   $128\le n \le 128$    &                  Extend $C_{31}$                  \\
			\textbf{31} & $\bm{C_{33}}$ & $\bm{[n,3.5,n-33]_4}$ & $\bm{129\le n \le 133}$ &             $\bm{(C_{8}\mid C_{25})}$             \\
			32      & $\bm{C_{34}}$ & $\bm{[n,3.5,n-34]_4}$ & $\bm{134\le n \le 138}$ &             $\bm{(C_{8}\mid C_{26})}$             \\
			\textbf{33} & $\bm{C_{35}}$ & $\bm{[n,3.5,n-35]_4}$ & $\bm{139\le n \le 141}$ &            $\bm{(C_{10}\mid C_{25})}$             \\
			\textbf{34} & $\bm{C_{36}}$ & $\bm{[n,3.5,n-36]_4}$ & $\bm{142\le n \le 146}$ &            $\bm{(C_{10}\mid C_{26})}$             \\
			\textbf{35} & $\bm{C_{37}}$ & $\bm{[n,3.5,n-37]_4}$ & $\bm{147\le n \le 149}$ &            $\bm{(C_{11}\mid C_{26})}$             \\
			\textbf{36} & $\bm{C_{38}}$ & $\bm{[n,3.5,n-38]_4}$ & $\bm{150\le n \le 154}$ &             $\bm{(C_{8}\mid C_{30})}$             \\
			37      &   $C_{39}$    &   $[n,3.5,n-39]_4$    &   $155\le n \le 159$    &               $(C_{8}\mid C_{31})$                \\
			38      &   $C_{40}$    &   $[n,3.5,n-40]_4$    &   $160\le n \le 162$    &               $(C_{9}\mid C_{31})$                \\
			39      &   $C_{41}$    &   $[n,3.5,n-41]_4$    &   $163\le n \le 167$    &               $(C_{10}\mid C_{31})$               \\
			40      &   $C_{42}$    &   $[n,3.5,n-42]_4$    &   $168\le n \le 170$    &               $(C_{11}\mid C_{31})$               \\
			\textbf{41} & $\bm{C_{43}}$ & $\bm{[n,3.5,n-43]_4}$ & $\bm{171\le n \le 173}$ &             $\bm{(C_{8}\mid C_{35})}$             \\
			\textbf{42} & $\bm{C_{44}}$ & $\bm{[n,3.5,n-44]_4}$ & $\bm{174\le n \le 178}$ &             $\bm{(C_{8}\mid C_{36})}$             \\
			\textbf{43} & $\bm{C_{45}}$ & $\bm{[n,3.5,n-45]_4}$ & $\bm{179\le n \le 181}$ &            $\bm{(C_{14}\mid C_{31})}$             \\
			\textbf{44} & $\bm{C_{46}}$ & $\bm{[n,3.5,n-46]_4}$ & $\bm{182\le n \le 186}$ &            $\bm{(C_{15}\mid C_{31})}$             \\
			45      &   $C_{47}$    &   $[n,3.5,n-47]_4$    &   $187\le n \le 191$    &               $(C_{16}\mid C_{31})$               \\
			46      &   $C_{48}$    &   $[n,3.5,n-48]_4$    &   $192\le n \le 194$    &               $(C_{17}\mid C_{31})$               \\
			47      &   $C_{49}$    &   $[n,3.5,n-49]_4$    &   $195\le n \le 199$    &               $(C_{18}\mid C_{31})$               \\
			48      &   $C_{50}$    &   $[n,3.5,n-50]_4$    &   $200\le n \le 202$    &               $(C_{19}\mid C_{31})$               \\
			49      &   $C_{51}$    &   $[n,3.5,n-51]_4$    &   $203\le n \le 207$    &               $(C_{20}\mid C_{31})$               \\
			50      &   $C_{52}$    &   $[n,3.5,n-52]_4$    &   $208\le n \le 212$    &               $(C_{21}\mid C_{31})$               \\
			51      &   $C_{53}$    &   $[n,3.5,n-53]_4$    &   $213\le n \le 213$    &               $(C_{22}\mid C_{31})$               \\
			\textbf{52} & $\bm{C_{54}}$ & $\bm{[n,3.5,n-54]_4}$ & $\bm{214\le n \le 218}$ &            $\bm{(C_{23}\mid C_{31})}$             \\
			53      &   $C_{55}$    &   $[n,3.5,n-55]_4$    &   $219\le n \le 223$    &               $(C_{24}\mid C_{31})$               \\
			\textbf{54} & $\bm{C_{56}}$ & $\bm{[n,3.5,n-56]_4}$ & $\bm{224\le n \le 228}$ &            $\bm{(C_{25}\mid C_{31})}$             \\
			\textbf{55} & $\bm{C_{57}}$ & $\bm{[n,3.5,n-57]_4}$ & $\bm{229\le n \le 233}$ &            $\bm{(C_{26}\mid C_{31})}$             \\
			58      &   $C_{60}$    &   $[n,3.5,n-58]_4$    &   $234\le n \le 234$    &               $(C_{27}\mid C_{31})$               \\
			59      &   $C_{61}$    &   $[n,3.5,n-59]_4$    &   $235\le n \le 239$    &               $(C_{28}\mid C_{31})$               \\
			\textbf{60} & $\bm{C_{62}}$ & $\bm{[n,3.5,n-60]_4}$ & $\bm{240\le n \le 244}$ &            $\bm{(C_{29}\mid C_{31})}$             \\
			\textbf{61} & $\bm{C_{63}}$ & $\bm{[n,3.5,n-61]_4}$ & $\bm{245\le n \le 249}$ &            $\bm{(C_{30}\mid C_{31})}$             \\
			62      &   $C_{64}$    &   $[n,3.5,n-62]_4$    &   $250\le n \le 254$    &               $(C_{31}\mid C_{31})$               \\ \bottomrule
		\end{tabular}%
		\begin{tablenotes}    %这行要添加， 从这开始
			\footnotesize               %这行要添加
			\item[]\textbf{Note:} The marked bold indicates the new quaternary additive codes obtained in this paper.
		\end{tablenotes}            %这行要添加
	\end{threeparttable}  
\end{table*}

	%%%%%%%%%%%%%%%%%-V-Conclusion--%%%%%%%%%%%%%%%%%%%%%%%%%%%%%%
	\section{Conclusion}\label{VII Dis conclu}
%	In this work, we present efficient combinatorial constructions for optimal quaternary additive codes and obtain many classes of optimal additive codes with non-integer dimensions. 
%	Specifically, the combinatorial construction of additive constant-weight codes is given, and the algebra structure of its generator matrix is characterized. 
%	Based on this, we give an anticode construction of additive codes. In addition, generalized Construction X is proposed, which allows the construction of additive codes of non-integer dimensions using special optimal linear codes.	
%	It is worth noting that Lemma \ref{outperform_linear} determines a sufficient condition for additive codes with non-integer dimensions to outperform linear codes, and most of our additive codes outperform optimal linear codes in \cite{Grassltable,Marutatables}, and we summarize them in Table \ref{T_Minimal_additive}. 
%Finally, we also determine optimal additive $[n,3.5,n-t]_4$ codes for all $t$ with variable $n$, except for $t=6,7,12,13$. 

In this work, we studied combinatorial techniques for quaternary additive codes and successfully derived ten classes of optimal additive codes with non-integer dimensions, see Table \ref{T_Minimal_additive}. 
%Notably, Lemma \ref{outperform_linear} established a sufficient condition for additive codes with non-integer dimensions to surpass linear codes in performance, which shows ten classes of our additive codes outperform linear counterparts \cite{Grassltable,Marutatables}, as listed in Table \ref{T_Minimal_additive}
Specifically, we detailed the combinatorial construction of additive constant-weight codes and elucidated the algebraic structure of their generator matrices, see Lemmas \ref{APS_construction}, \ref{iterating} and \ref{L_camba}.
Then, we proposed an additive generalized anticode construction method for additive codes, see Lemma \ref{anticode construction} and Theorem  \ref{anticode_k_1_k_2}. Furthermore, we present a generalized Construction X that enables the construction of additive codes with non-integer dimensions using specialized optimal linear codes, see Lemma \ref{Generalized_X_Construction_Corollary} and Theorem \ref{Generalized_X_Construction}.
Finally, we also determined the optimal additive $[n,3.5,n-t]_4$ codes for all $t$ with variable $n$, except for $t=6,7,12$, see Lemma \ref{Additive code combinations} and Table \ref{3.5_additive}.

%In particular, our results can cover recent works about optimal additive codes in  \cite{bierbrauer2021optimal} and \cite{Guan2023SomeGQ}, and we show this in Remark \ref{cover_results}.

%Therefore, in the future, the decoding of additive codes and their applications in cryptography and other aspects will be interesting problems.

%\section*{Acknowledgments}
%	%The authors would like to thank the discussions with Dr. Shitao Li and Dr. Yang Li, whose suggestions improved the quality of this paper.
%	The authors would like to thank Dr. Shitao Li and Dr. Yang Li for the helpful discussions. 
%	%They also thank the anonymous reviewers for their valuable comments.
%	They also thank three anonymous reviewers and the Associate Editor Prof. Xiande Zhang for their valuable comments and suggestions.
% They also thank the reviewers and  for carefully reading this paper and valuable comments, which significantly improved this paper's presentation and quality.
	%
	%The authors are very grateful to Prof. Markus Grassl for discussing the relationship between additive and quasi-cyclic codes, which significantly improved this paper's presentation and quality.

	%%%%%%%%%%%%%%%%%%%%%%%%%%%%%%%%%%%%%%%%%%%%%%%%%%%%%%%%%%%%%%%%%%%%%%%%%
	\bibliographystyle{IEEEtran} 
	\bibliography{reference} 

\end{document}